\documentclass[11pt]{article}

\title{\textbf{Planar Graphs: Random Walks and Bipartiteness Testing}\thanks{To appear in \emph{Random Structures and Algorithms}, {\tt http://doi.org/10.1002/rsa.20826}\,, license: CC-BY. A preliminary version appeared in \emph{Proceedings of the 52th IEEE Symposium on Foundations of Computer Science (FOCS)}, pages 423--432, Palm Springs, CA, October 22--25, 2011. IEEE Computer Society Press, Los Alamitos, CA, 2011.}}
\author{Artur Czumaj\thanks{Department of Computer Science, Centre for Discrete Mathematics and its Applications, University of Warwick. Email: A.Czumaj@warwick.ac.uk. Research supported in part by the \emph{Centre for Discrete Mathematics and its Applications (DIMAP)}, EPSRC award EP/D063191/1, by EPSRC award EP/G064679/1, and by a Weizmann-UK Making Connections Grant ``The Interplay between Algorithms and Randomness.''}
    \and
		Morteza Monemizadeh\thanks{Amazon, Palo Alto, CA, USA and Computer Science
Institute of Charles University, Faculty of Mathematics and Physics,
Prague, Czech Republic. The work was done when the author was at the Department of
Computer Science, Goethe-Universit{\"a}t Frankfurt, Germany, University of
Maryland, College Park, MD, USA and Computer Science Institute of Charles
University, Faculty of Mathematics and Physics, Prague, Czech Republic.
Email: monemi@iuuk.mff.cuni.cz, mmorteza@amazon.com.
Partially supported by the projects MO 2200/1-1 and 14-10003S of GA \v{C}R.
}
    \and
Krzysztof Onak\thanks{IBM T.J.\ Watson Research Center. Part of the work was done when the author was at the Massachusetts Institute of Technology and Carnegie Mellon University. Email: konak@us.ibm.com. Supported in part by a Simons Postdoctoral Fellowship and NSF grants 0732334 and 0728645.}
    \and
Christian Sohler\thanks{Department of Computer Science, TU Dortmund. Email: christian.sohler@tu-dortmund.de. Supported in part by the German Research Foundation (DFG), grant So 514/3-2 and ERC grant No.\ 307696.}
}

\date{August 2018}

\usepackage[T1]{fontenc}\usepackage{yfonts}
\usepackage{amsmath,amssymb,amsfonts}
\usepackage{graphicx}
\usepackage{subfigure}
\usepackage{latexsym}
\usepackage{graphicx}
\usepackage{times}
\usepackage{fancybox}
\usepackage{fullpage}
\usepackage{color}
\usepackage{paralist}
\usepackage{environ}
\usepackage{xstring}

\usepackage[T1]{fontenc}
\usepackage{yfonts}

\usepackage{ifpdf}
\newcommand{\mydriver}{hypertex}
\ifpdf
 \renewcommand{\mydriver}{pdftex}
\fi
\usepackage[breaklinks,\mydriver]{hyperref}

\DeclareMathOperator{\poly}{poly}

\newenvironment{proof}{\noindent {\bf Proof}.\ }{\qed \par\vskip 4mm\par}

\newtheorem{theorem}{Theorem}

\newtheorem{lemma}[theorem]{Lemma}

\newtheorem{defn}[theorem]{Definition}
\newenvironment{definition}[1][]{\IfStrEq{#1}{}{\begin{defn}}{\begin{defn}[#1]}\rm}{\end{defn}}

\newtheorem{claim}[theorem]{Claim}

\newtheorem{fact}[theorem]{Fact}

\newcommand{\qed}{\hspace*{\fill}\sq}
\renewcommand{\qed}{\hspace*{\fill}\ensuremath{\blacksquare}}
\newcommand{\COMMENTED}[1]{{}}
\newcommand{\junk}[1]{\COMMENTED{#1}}

\newcommand{\NAT}{\ensuremath{\mathbb{N}}}

\newcommand{\eps}{\ensuremath{\epsilon}}
\def\epsilon{\ensuremath{\varepsilon}}

\newlength{\savedparindent}
\newcommand{\SaveIndent}{\setlength{\savedparindent}{\parindent}}
\newcommand{\RestoreIndent}{\setlength{\parindent}{\savedparindent}}

\newcommand{\InGray}[1]{%
\SaveIndent{} %
\noindent{} \fcolorbox[rgb]{0,0,0}{0.95,0.95,0.95}{
\begin{minipage}{0.965\linewidth} %
\RestoreIndent{}%
#1
\end{minipage}
} }

\newcommand{\RBE}{{\bf Random-Bipartiteness-Exploration}}
\newcommand{\RW}{{\bf Random-Walk}}

\newcommand{\AL}{\textbf{Assigning-Levels}}

\NewEnviron{algo}{\smallskip\InGray{\mbox{}\\\BODY\mbox{}\\[-0.35in]}\medskip}

\begin{document}

\maketitle

\begin{abstract}
We initiate the study of property testing in \emph{arbitrary planar graphs}. We prove that \emph{bipartiteness} can be tested in constant time, improving on the previous bound of $\tilde{O}(\sqrt{n})$ for graphs on $n$ vertices. The constant-time testability was only known for planar graphs with \emph{bounded degree}.

Our algorithm is based on random walks. Since planar graphs have good separators, i.e., bad expansion, our analysis diverges from standard techniques that involve the fast convergence of random walks on expanders. We reduce the problem to the task of detecting an odd-parity cycle in a multigraph induced by constant-length cycles. We iteratively reduce the length of cycles while preserving the detection probability, until the multigraph collapses to a collection of easily discoverable self-loops.

Our approach extends to arbitrary minor-free graphs. We also believe that our techniques will find applications to testing other properties in arbitrary minor-free graphs.
\end{abstract}


\section{Introduction}

\emph{Property testing} studies relaxed decision problems in which one wants to distinguish objects that have a given property from those that are far from this property (see, e.g., \cite{Gol10}). Informally, an object $\mathcal{X}$ is $\eps$-far from a property $\mathcal{P}$ if one has to modify at least an $\eps$-fraction of $\mathcal{X}$'s representation to obtain an object with property $\mathcal{P}$, where $\eps$ is typically a small constant. Given oracle access to the input object, a typical property tester achieves this goal by inspecting only a small fraction of the input.
Property testing is motivated by the need to understand how to extract information efficiently from massive structured or semi-structured data sets using small (possibly adaptive) random samples.

One of the main and most successful directions in property testing is \emph{testing graph properties}, as introduced in papers of Goldreich et al.\ \cite{GGR98,GR97}. There are two popular models for this task, which make different assumptions about how the input graph is represented and how it can be accessed.

For a long time, the main research focus has been on the \emph{adjacency matrix model}, designed specifically for \emph{dense} graphs \cite{GGR98}. In this model, after a sequence of papers, it was shown that testability of a property in constant time is closely related to Szemer\'{e}di partitions of the graph. More precisely, a property is testable in constant time\footnote{Throughout the paper we say that a property is \emph{testable in constant time} if there is a testing algorithm whose number of queries to the input is independent of the input size, possibly depending only on the proximity parameter $\eps$.} if and only if it can be reduced to testing finitely many Szemer\'{e}di partitions \cite{AFNS06}.

The \emph{general graph model} or incidence list model has been introduced in \cite{PR02} and assumes that the graph is stored using incidence lists. Each list contains the vertex degree (length of the list) at the beginning, i.e., one can access the vertex degree in constant time.
A variant of this model where the algorithm can specify a pair of vertices and query whether they are adjacent has been introduced in \cite{KKR04}.

 In the \emph{bounded-degree graph model} introduced in \cite{GR97} we have the additional restriction that the degree of the graph is at most a certain predefined constant $d$. Unlike in the adjacency matrix model, it is not yet completely understood what graph properties are testable in constant time in the adjacency list model. If the underlying graph has bounded degree and is planar, Czumaj et al.\ \cite{CSS09} show that any hereditary graph property\footnote{A graph property is \emph{hereditary} if it is closed under vertex removals.} is testable in constant time. This approach can be generalized to any class of graphs that can be partitioned into constant-size components by removing $\eps n$ edges of the graph, for any $\eps > 0$ (the technical condition in the paper was formulated slightly differently). Graphs satisfying this property are called \emph{hyperfinite}, and they include all bounded-degree minor-closed graph families.
A sequence of papers \cite{BSS08,HKNO09} led to the result that all hyperfinite properties are testable \cite{NS11}, i.e., for a general bounded-degree input graph
any property that consists only of hyperfinite graphs can be tested in constant time. Other testable properties include connectivity, $k$-edge-connectivity, the property of being Eulerian \cite{GR97}, and the property of having a perfect matching \cite{NO08}. Furthermore, every property of a certain class of scale-free
multigraphs is testable \cite{I16}.
On the other hand, some properties testable in constant time in the dense graph model, such as bipartiteness and 3-colorability, are known to require a superconstant number of queries \cite{BOT02,GR97}.

Most of the positive results mentioned above (in particular, \cite{CSS09,NS11,BSS08,HKNO09}) are heavily based on the fact that a bounded-degree hyperfinite
graph has a small edge separator. This edge separator allows to partition the graph into pieces of constant size by removing $\epsilon d n$ edges and can be thought of as reducing property testing in a large graph to a graph that consists of a collection of small components. Furthermore, most constant time testers work by performing a BFS from a constant size set of sample vertices and decide based on the sampled graph.

If the graph does not have a degree bound, such an approach can no longer work. First of all, general planar graphs do not necessarily have a small edge separator (consider, for example, a star). Furthermore, a BFS up to constant depth can no longer be performed in constant time and therefore, a different algorithmic approach
has to be considered. Overall, much less is known about efficiently testable properties for sparse graphs that do not have a degree bound.
It is known that connectivity, $k$-edge-connectivity, and Eulerian graphs are testable in constant time \cite{MR06}. Recently,
some more general results were obtained for simple classes of graphs such as trees \cite{KY14} and outerplanar graphs \cite{BKN16}.

\paragraph{Our contribution.}
The central goal of the this paper is to initiate the research on the complexity of testing graph properties in
general unbounded degree minor-closed graph families. Furthermore, the inquiry into the complexity of property testing for bounded-degree minor-closed (or, more generally, hyperfinite) input graphs was an important step towards our current understanding of property testing in the bounded-degree graph model.

We hope that following a similar route, we will in the long run also reach a better understanding of testable properties in graphs with unbounded degree.
The technical contribution in this paper is supposed to be a first step in this direction. We develop a new analysis of a random walk approach, which
was first introduced in \cite{GR99} for testing bipartiteness in arbitrary bounded-degree graphs. We
illustrate the usefulness of our approach by giving a proof that bipartiteness is testable in constant time.

Similarly to the case of bounded-degree minor-closed graph families, our approach exploits a form of the graph separator theorem.
However, in the case of unbounded degree, only
a weaker form of the planar separator theorem is available, which allows us to partition the graph into subgraphs of small diameter by removing an $\epsilon$-fraction of the edges. In the next step we would like to argue that for every graph that is $\epsilon$-far from the tested property, the obtained partition
classes contain small counter-examples to the tested graph properties, i.e., in our showcase of testing bipartiteness, we prove that there is a large set of
short odd-length cycles.

The main contribution of the paper is then to show that a random walk from a random starting vertex finds such an odd-length cycle with small constant probability, so that repeating this walk a constant number of times will result in a property tester. In order to show that our algorithm works, we design
a reduction that takes an input graph $G$ with a large set of edge-disjoint odd-length cycles of length at most $k$ and reduce it to a graph $G'$ with the
properties that
\begin{itemize}
\item
$G'$ contains a large set of edge-disjoint odd-length cycles of length at most $k-1$ and
\item
if a random walk finds with constant probability an odd-length cycle in $G'$ then it does so in $G$.
\end{itemize}
We repeatedly apply this reduction until we obtain trivial odd-length cycles, i.e., self-loops that are easy to detect.

While our analysis at places uses the simple structure of our forbidden subgraph, i.e., that it is an odd-length cycle, it seems to be likely that a
similar reduction can work to find more complex subgraphs as well (however, there are also non-trivial obstacles for this to happen). At the same time
we remark that the technical details of our reduction are already highly non-trivial for the case of testing bipartiteness.


\paragraph{Bipartiteness.}

The problem of testing bipartiteness has been a great benchmark of the capabilities of property testing algorithms in various graph models. It was one of the first problems studied in detail in both the dense graph model \cite{GGR98} and the bounded-degree graph model \cite{GR97,GR99}. Bipartiteness is known to be testable in $\widetilde{O}(1/\eps^2)$ time in the dense graph model \cite{AK02}. However, in the bounded-degree graph model, it requires $\Omega(\sqrt{n})$ queries \cite{GR97} and is testable in $\widetilde{O}(\sqrt{n} \cdot \eps^{-O(1)})$ time \cite{GR99}, where $n$ is the number of vertices. Kaufman et al.\ \cite{KKR04} show that the property is still testable in $\widetilde{O}(\sqrt{n} \cdot \eps^{-O(1)})$ time in the adjacency list model for graphs that have constant \emph{average} degree.



\subsection{Approaches that do not work}

Given that bipartiteness can be tested in constant time in planar graphs of bounded degree \cite{CSS09}, it may seem that there is a simple extension of this result to arbitrary degrees. We now describe two natural attempts at reducing our problem to testing bipartiteness in other classes of graphs. We explain why they fail. We hope that this justifies our belief that new techniques are necessary to address the problem.


\begin{figure}[t]
\centerline{
(a)
\includegraphics[width=0.40\textwidth]{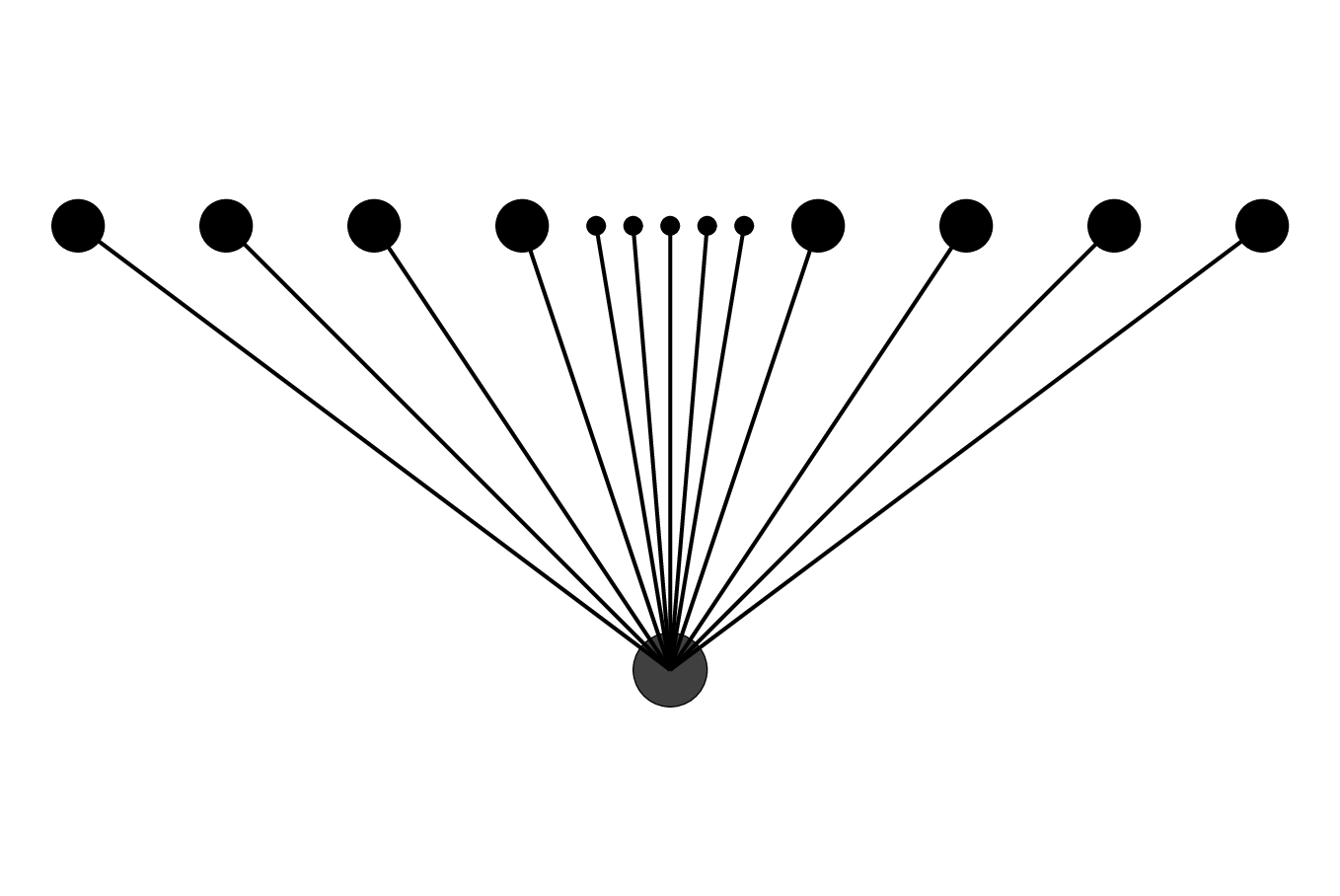}
        \qquad\quad
(b)
\includegraphics[width=0.40\textwidth]{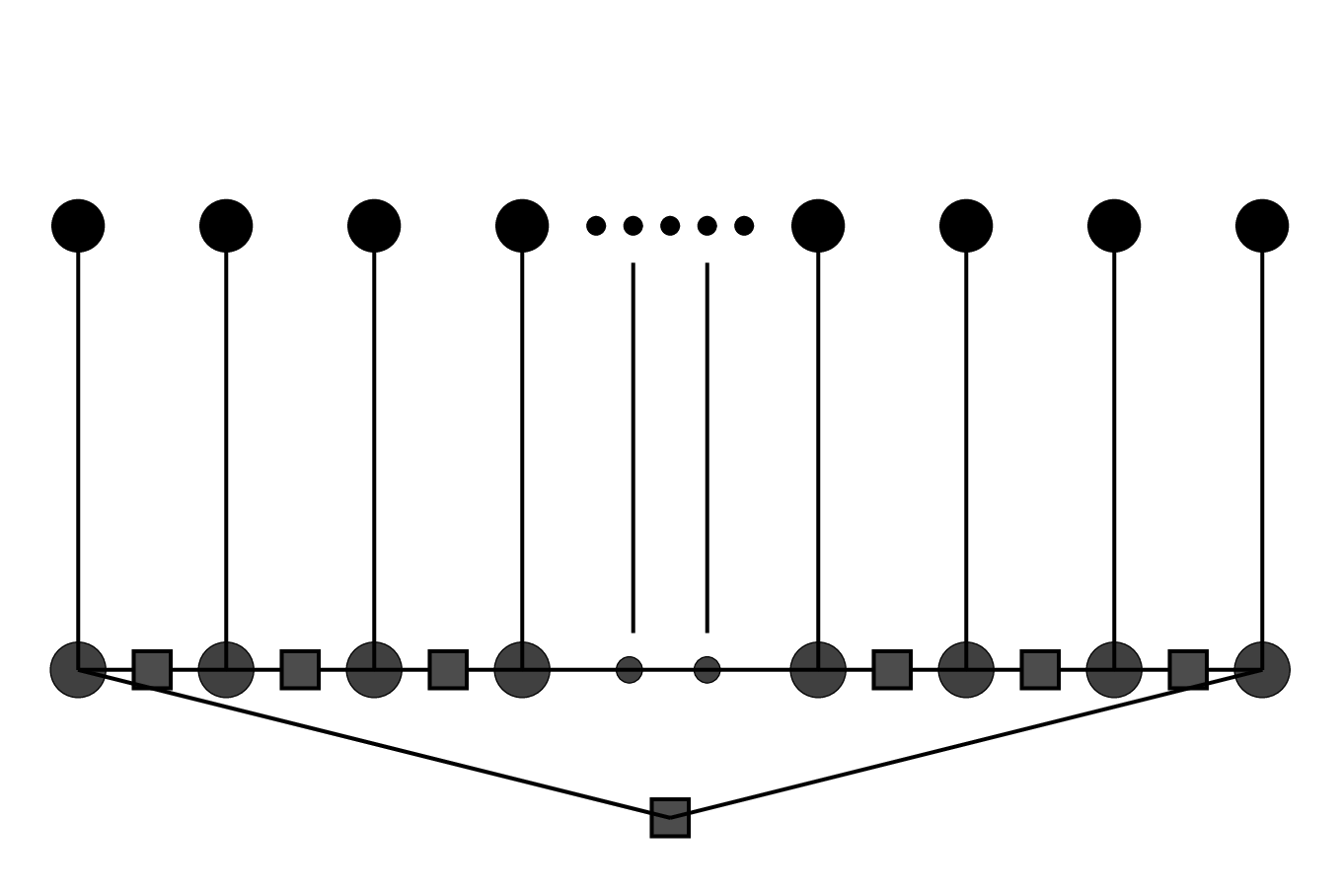}
}
\caption{An example of the process of splitting a vertex that reduces any graph into a graph of maximum degree at most $3$ and that maintains planarity. For the graph in (a), Figure (b) depicts the splitting that is invariant to being bipartite.}
\label{fig:splitting-to-deg-3}
\end{figure}

\begin{figure}[t]
\centerline{
(a)
\includegraphics[width=0.25\textwidth]{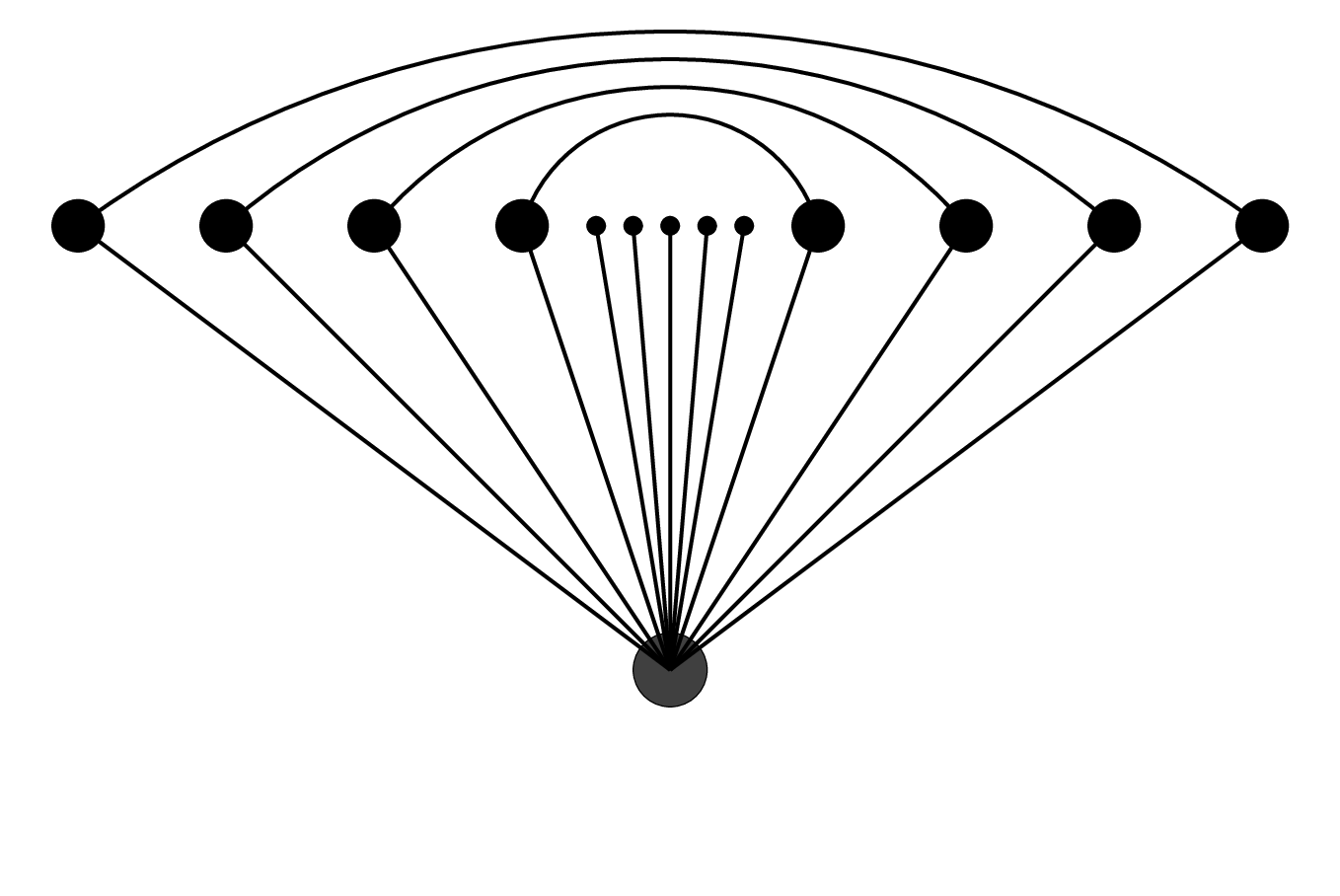}
        \qquad\quad
(b)
\includegraphics[width=0.25\textwidth]{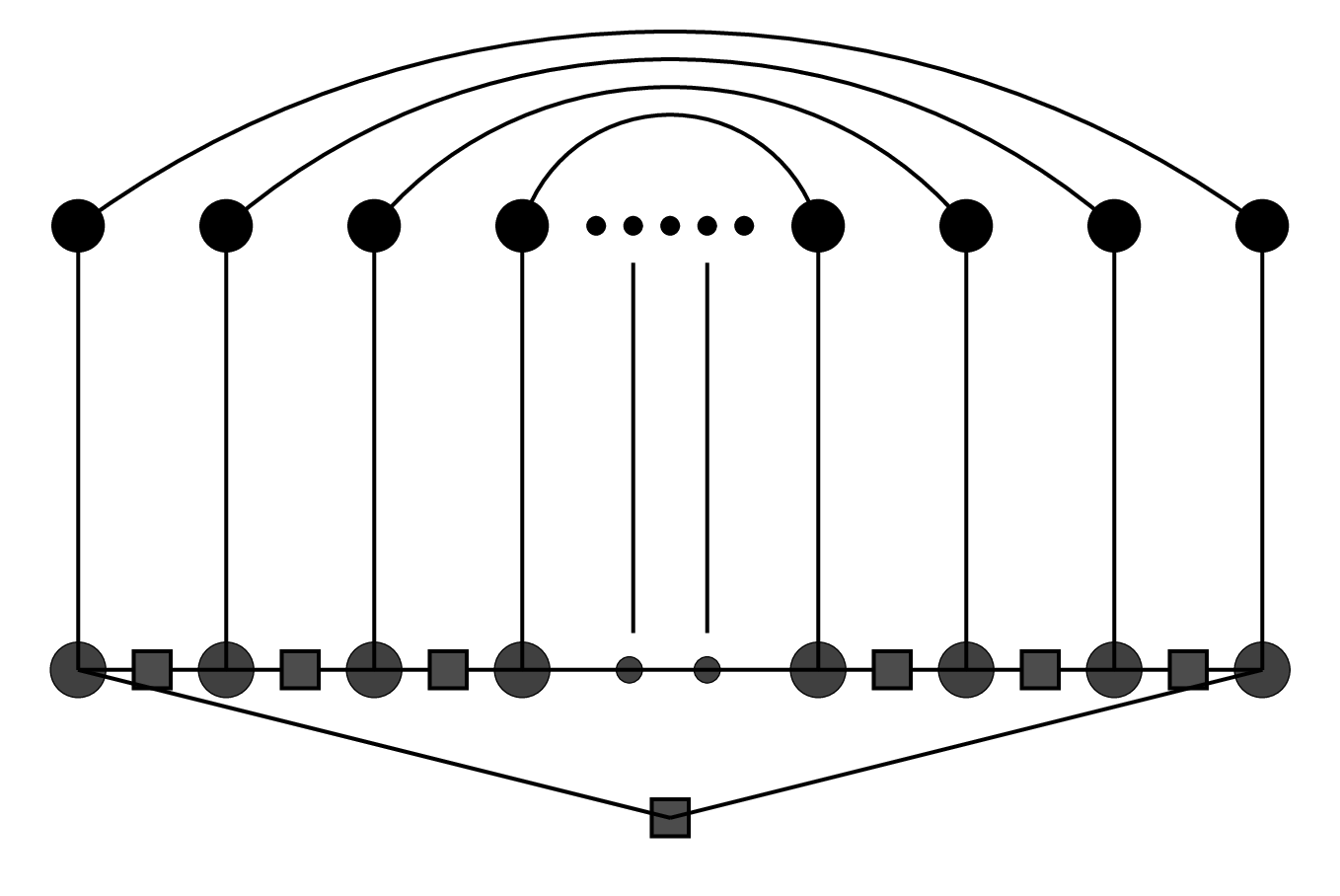}
        \qquad\quad
(c)
\includegraphics[width=0.25\textwidth]{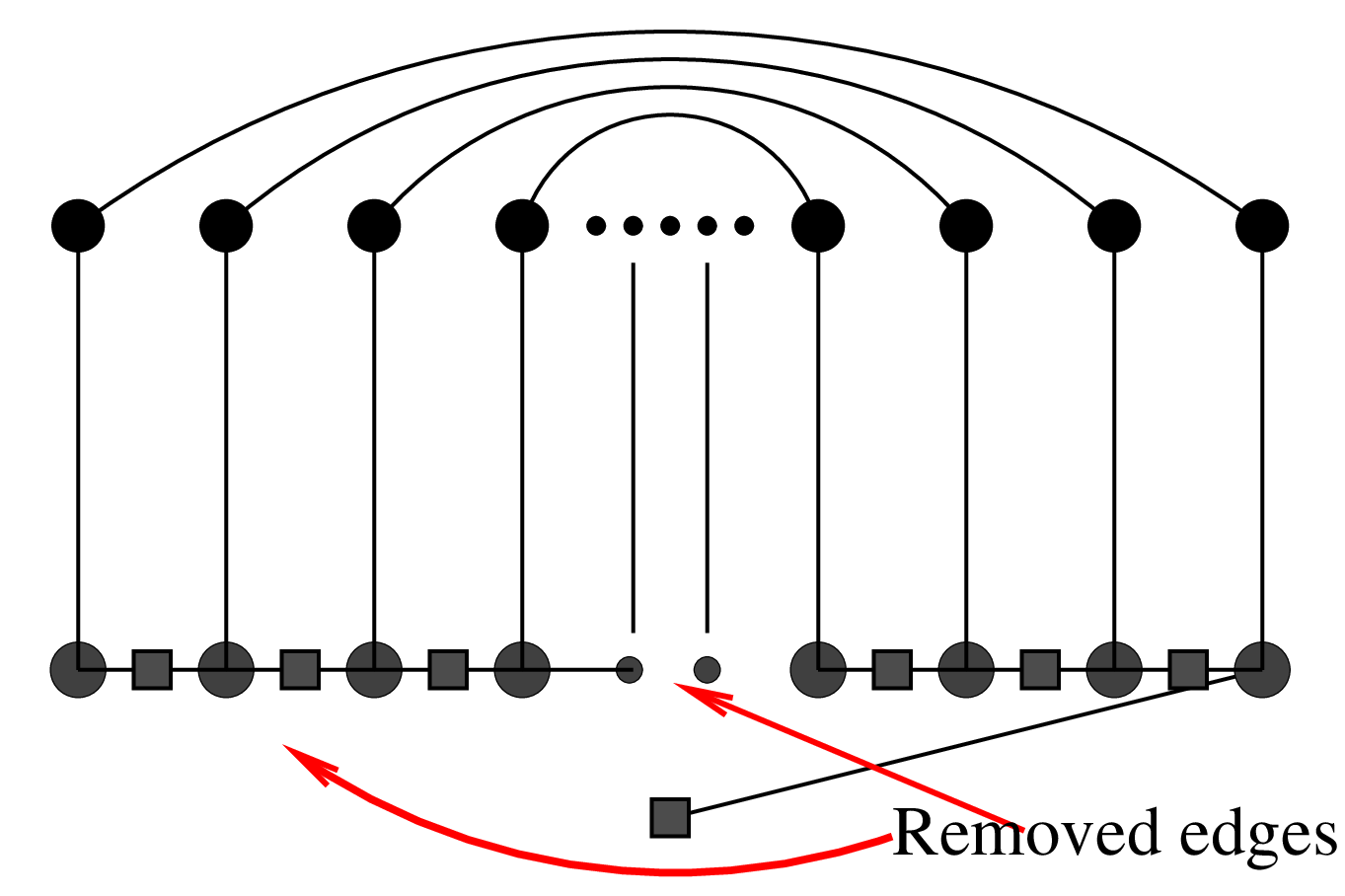}
}
\caption{An example showing that the splitting construction from Figure \ref{fig:splitting-to-deg-3} can reduce the distance from being bipartite. The planar graph in (a) (in which the $i^{\text{th}}$ top vertex from the left is connected by an edge to the $i^{\text{th}}$ top vertex from the right) has $\Theta(n)$ edge-disjoint cycles of length $3$ and is $\eps$-far from bipartite (one has to remove at least $\frac{n-1}{2}$ edges to obtain a bipartite graph). However, after the splitting, the obtained graph (Figure (b)) can be made bipartite just by removal of two edges: Figure (c) depicts a bipartite graph obtained after removal of such two edges: one of the two edges at the bottom and the middle edge in the split part.}
\label{fig:splitting-to-deg-3-with-far}
\end{figure}


The first and possibly the most natural approach to designing a constant-time algorithm for testing bipartiteness in arbitrary planar graphs would be to extend the known constant-time algorithm for \emph{bounded-degree} planar graphs \cite{CSS09}. This could be achieved by first transforming an input planar graph $G$ with an arbitrary maximum-degree into a planar graph $G^\star$ with bounded-degree and then running the tester for $G^\star$ to determine the property for $G$. However, we are not aware of a transformation that would behave well and we do not expect any such transformation to exist.

For example, one can reduce the maximum degree to at most $3$ by splitting every vertex of degree $d > 3$ into $d$ vertices of degree $3$. It is also easy to ensure that this reduction maintains the planarity, and also the property of being bipartite (see Figure \ref{fig:splitting-to-deg-3}). However, there are two properties that are not maintained: one is the distance from being bipartite (see Figure \ref{fig:splitting-to-deg-3-with-far}) and another is that the access to the neighboring nodes requires more than constant time (though this can be ``fixed'' if one allows each vertex to have its adjacency list ordered consistently with some planar embedding). In particular, Figure \ref{fig:splitting-to-deg-3-with-far} depicts an example of a planar graph that is originally $\eps$-far from bipartite, but after the transformation it suffices to remove 2 edges to obtain a bipartite graph.


Another transformation of the graph is considered by Kaufman et al.\ \cite{KKR04}. They replace every high degree vertex with a constant-degree bipartite expander. While they prove that this construction preserves the distance, it is clear that it cannot preserve the planarity, since planar graphs are not expanders. However, for general graphs we know that testing bipartiteness requires $\Omega(\sqrt{n})$ queries \cite{GR97} and we do not know how to exploit the structure of the graph after the transformation.


\section{Preliminaries}
\label{sec:preliminaries}

\paragraph{Bipartiteness.}

A graph is \emph{bipartite} if one can partition its vertex set into two sets $A$ and $B$ such that every edge has one endpoint in $A$ and one endpoint in $B$. We also frequently use the well known fact that a graph is bipartite if and only if it has no odd-length cycle.

We now formally introduce the notion of being far from bipartiteness.\footnote{The standard definition of being $\eps$-far (see, for example, the definition in \cite{KKR04}) expresses the distance as the fraction of edges that must be modified in $G=(V,E)$ to obtain a bipartite graph. Compared to our Definition \ref{def:bipartiteness-testing}, instead of deleting $\eps |V|$ edges, one can delete $\eps |E|$ edges. For any class of graphs with an excluded minor, the number of edges in the graph is upper bounded by $C \cdot |V|$, where $C$ is a constant. Moreover, unless the graph is very sparse (i.e., most of its vertices are isolated, in which case even finding a single edge in the graph may take a large amount of time), the number of edges in the graph is at least $\Omega(|V|)$. Thus, under the standard assumption that $|E| = \Omega(|V|)$, the $\eps$ in our definition and the $\eps$ in the previous definitions remain within a constant factor. We use our definition of being $\eps$-far for simplicity; our analysis can be extended to the standard definition of being $\eps$-far in a straightforward way.} The notion is parameterized by a distance parameter $\eps > 0$.

\begin{definition}
\label{def:bipartiteness-testing}
A graph $G=(V,E)$ is \emph{$\eps$-far from bipartite} if one has to delete more than $\eps |V|$ edges from $G$ to obtain a bipartite graph.
\end{definition}


\paragraph{Property testing.}

We are interested in finding a \emph{property testing algorithm} for bipartiteness in planar graphs, i.e., an algorithm that inspects only a very small part of the input graph, and accepts bipartite planar graphs with probability at least $\frac23$, and rejects planar graphs that are $\eps$-far away from bipartite with probability at least $\frac23$, where $\eps$ is an additional parameter.

Our algorithm always accepts every bipartite graph. Such a property testing algorithm is said to have \emph{one-sided error}.


\paragraph{Access model.}
The access to the graph is given by an \emph{oracle}. We consider the oracle that allows two types of queries:
\begin{itemize}
\item {\it Degree queries:} For every vertex $v \in V$, one can query the degree of $v$.
\item {\it Neighbor queries:} For every vertex $v \in V$, one can query its $i^{\text{th}}$ neighbor.
\end{itemize}
Observe that by first querying the degree of a vertex, we can always ensure that the $i^{\text{th}}$ neighbor of the vertex exists in the second type of query. In fact, in the algorithm that we describe in this paper, the neighbor query can be replaced with a weaker type of query: \emph{random neighbor query}, which returns a random neighbor of a given vertex $v$; each time the neighbor is chosen independently and uniformly at random.

The \emph{query complexity} of a property testing algorithm is the number of oracle queries it makes.


\paragraph{Basic properties of planar graphs.}

We extensively use the following well-known properties of planar graphs. The graph $G' = (V',E')$ obtained by the \emph{contraction} of an edge $(u,v) \in E$ into vertex $u$ is defined as follows: $V' = V \setminus \{v\}$ and $E' = \{(x,y)\in E: x \ne v \land y \ne v\} \cup \{(x,u) : (x,v) \in E \land x \ne u\}$. A graph $G'$ that can be obtained from a graph $G$ via a sequence of edge removals, vertex removals, and edge contractions is called a \emph{minor} of $G$.

We use the following well-known property of planar graphs.

\begin{fact}
\label{fact:planar_minor_planar}
Any minor of a planar graph is planar.
\end{fact}

Furthermore, we use the following upper bound on the number of edges in a simple planar graph, which follows immediately from Euler's formula.

\begin{fact}
\label{fact:planar_limited_edges}
For any simple planar graph $G = (V,E)$ (with no self-loops or parallel edges), $|E| \le 3|V|-6$.
\end{fact}

We remark that for any class of graphs $\mathcal{H}$ that is defined by a finite collection of forbidden minors, similar statements are true, i.e., if $G \in \mathcal{H}$, then any minor of $G$ also belongs to $\mathcal{H}$ and if $G = (V,E) \in \mathcal{H}$, then $G$ has $O(|V|)$ edges (where the constant in the big $O$ notation depends on the set of forbidden minors).


\paragraph{Notation.}

Throughout the paper we use several constants depending on $\eps$. We use lower case Greek letters to denote constants that are typically smaller than $1$ (e.g., $\delta_i(\eps)$) and lower case Latin letters to denote constants that are usually larger than $1$ (e.g., $f_i(\eps)$). All these constants are always positive.
Furthermore, throughout the paper we use the asymptotic symbols $O_{\eps}(\cdot)$, $\Omega_{\eps}(\cdot)$, and $\Theta_{\eps}(\cdot)$, which ignore multiplicative factors that depend only on $\eps$ and that
are positive for $\eps>0$.


\section{Algorithm \RBE}

We first describe our algorithm for testing bipartiteness of planar graphs with arbitrary degree and provide the high level structure of its analysis. Most of the technical details appear in Section \ref{sec:main-proof}.

\medskip
\begin{algo}
\RBE\,$(G,\eps)$:
\begin{itemize}
\item Repeat $f(\eps)$ times:
    \begin{itemize}[$\circ$]
    \item Pick a random vertex $v \in V$.
    \item Perform a random walk of length $g(\eps)$ from $v$.
    \item If the random walk found an odd-length cycle, then \textbf{reject}.
    \end{itemize}
\item If none of the random walks found an odd-length cycle, then \textbf{accept}.
\end{itemize}
\end{algo}

\begin{theorem}
\label{thm:main-bipartiteness}
There are positive functions $f$ and $g$ such that for every planar graph $G$, we have
\begin{itemize}
\item if $G$ is bipartite, then \RBE$(G,\eps)$ accepts $G$, and
\item if $G$ is $\eps$-far from bipartite, then \RBE$(G,\eps)$ rejects $G$ with probability at least $0.99$.
\end{itemize}
\end{theorem}

We first observe that the first claim is obvious: if $G$ is bipartite, then every cycle in $G$ is of even length and hence \RBE{} always accepts. Thus, to prove Theorem~\ref{thm:main-bipartiteness}, it suffices to show that if $G$ is $\eps$-far from bipartite, then \RBE{} rejects $G$ with probability at least $0.99$. Therefore, from now on, we assume that the input graph $G$ is $\eps$-far from bipartite for some constant $\eps>0$. Furthermore, note that it suffices to show that a \emph{single} random walk of length $O_{\eps}(1)$ finds an odd-length cycle with probability $\Omega_{\eps}(1)$. Indeed, for any functions $g$ and $f$, if a random walk of length $g(\eps) = O_{\eps}(1)$ finds an odd-length cycle with probability at least $5/f(\eps) = \Omega_{\eps}(1)$, then this implies that $f(\eps)=O_{\eps}(1)$ independent random walks detect at least one odd-length cycle with probability at least $1-(1-5/f(\eps))^{f(\eps)} \ge 1-e^{-5} \ge 0.99$. Therefore, in the remainder of the paper, we analyze algorithm \RW$(G,t)$ below. We have to prove that there is $t=t(\eps)=O_{\eps}(1)$ such that for every planar graph $G$ that is $\eps$-far from bipartite, \RW$(G,t)$ finds an odd-length cycle with probability $\Omega_{\eps}(1)$. This implies Theorem \ref{thm:main-bipartiteness}.

\begin{algo}
\RW\,${(G,t)}$:
\begin{itemize}
\item  Pick a random vertex $v \in V$.
\item  Perform a random walk of length $t$ from $v$.
\item  If the random walk found an odd-length cycle, then \textbf{reject}.
\item  If not, then \textbf{accept}.
\end{itemize}
\end{algo}

\subsection{Outline of the analysis of \RW{} when $G$ is $\epsilon$-far from bipartite}
\label{subsec:outline}

Because of the arguments presented above, the remainder of the paper deals with the main technical challenge of our result: proving that our algorithm \RW{} finds with sufficient probability an odd-length cycle in any planar graph $G=(V,E)$ that is $\eps$-far from bipartite. To this end, we find a subgraph $H$ of $G$ that has the properties stated in the following lemma:

\begin{lemma}
\label{lemma:ExistenceOfH}
For every $\epsilon \in (0,1)$, there is a $t=t(\epsilon)>0$ such that
for every planar graph $G=(V,E)$ that is $\epsilon$-far from bipartite,
there exists a subgraph $H=(V,E')$, $E'\subseteq E$, of $G$ with the following properties:
\begin{itemize}
\item[(a)] if \ \RW{($H,t$)} finds an odd-length cycle in $H$ with probability $\Omega_{\eps}(1)$, then \RW{($G,t$)} finds an odd-length cycle in $G$ with probability $\Omega_{\eps}(1)$, and
\item[(b)] \RW{($H,t$)} finds an odd-length cycle in $H$ with probability $\Omega_{\eps}(1)$.
\end{itemize}
\end{lemma}

If such a subgraph $H$ always exists, these properties immediately imply that \RW{$(G,t)$} finds an odd-length cycle in $G$ with probability $\Omega_{\eps}(1)$ and so by the discussion above, Theorem~\ref{thm:main-bipartiteness} follows.

In order to prove the existence of $H$, we construct a series of subgraphs $G \supseteq H_1 \supseteq \dots \supseteq H_k$ for some $k=k(\epsilon)$. Each $H_i$, $1 \le i \le k$, satisfies property (a) and $H:=H_1$ also satisfies property (b). The entire series $H_1$, \dots, $H_k$ provides a proof that this is the case.

We begin by constructing graph $H_1$, which consists of a union of $\Omega_{\eps}(n)$ edge-disjoint short odd-length cycles from $G$. In order to construct $H_1$ we use the Klein-Plotin-Rao decomposition theorem, which shows that one can remove $\epsilon |V|/2$ edges from $G$ to partition $G$ into connected components such that the every pair of vertices from the same component has distance $O_{\eps}(1)$ in $G$. Then we show that there exists a component that contains a short odd-length cycle. We remove the cycle from the graph and repeat this process as long as we find a short odd-length cycle. The set of removed cycles forms a graph $H_1'$, which after some further processing to satisfy property (a) is turned into $H_1$. The detailed construction appears in Section \ref{subsec:constructing-H1}.

Then we design a reduction that takes a graph $H_i$ consisting of a set $C_i$ of $\Omega_{\eps}(n)$ edge-disjoint short odd-length cycles and constructs from it a subgraph $H_{i+1}$ that consists of a subset $C_{i+1}$ of $C_i$ of $\Omega_{\eps}(n)$ cycles. The new subgraph $H_{i+1}$ is supposed to approximately ''inherit'' the properties of the random walk from $H_i$.
To guide our process, we associate with each $H_i$ a certain multigraph $\mathcal H_i$ that is constructed from $H_i$ by performing edge contractions while keeping parallel edges and self-loops. The purpose of $\mathcal H_i$ is to approximate \emph{how the random walk behaves at a larger scale}, i.e., one step of a (weighted) random walk in $\mathcal H_i$ corresponds to multiple steps in $H_i$. Each vertex $u$ of $\mathcal H_i$ corresponds to a subset of vertices from $V$ that contains $u$. Furthermore, these subsets are disjoint and form a partition of $V$. Vertex $u$ from $\mathcal H_i$ can be thought of as a representative of the corresponding set. In particular, a random walk in $\mathcal H_i$ chooses its starting vertex $u$ with probability proportional to the size of its represented subset. Edges in $\mathcal H_i$ represent paths in $H_i$; the paths represented by parallel edges
may intersect. Parallel edges are also taken into account when we perform a random walk. The probability to move from $u$ to $v$ is proportional to the number of parallel edges $(u,v)$.

Let us recall that $H_i$ consists of a set of short odd-length cycles $C_i$. The construction of $\mathcal H_i$ ensures that every cycle in $C_i$ also corresponds to a cycle in $\mathcal H_i$. Clearly, the length of the corresponding cycle in $\mathcal H_i$ is at most the length of the cycle in $H_i$ but due to edge contractions it may become lower. When we are constructing $H_{i+1}$ from $H_i$ by removing cycles from $C_i$ to obtain $C_{i+1}$, we ensure that for every cycle from $C_{i+1}$ the length in $\mathcal H_{i+1}$ is shorter than the length of the cycle in $\mathcal H_i$ (in addition to preserving property (a)). Thus, after $k-1$ steps (when $k$ is the maximum initial length of a cycle) all cycles are contracted to self-loops. But then, if $\mathcal H_k$ properly approximates the behavior of the random walk at larger scale, a cycle will be detected. Furthermore, if the random walk properties of $H_{k-1}$, \dots, $H_1$ are approximately those of $H_k$, we obtain that $H_1$ satisfies both property (a) and (b), which finishes the proof.

While the above construction outlines the main line of thought, multiple details have to be taken care of. For instance, during the edge contractions we lose information about the original lengths of the paths and so an odd-length cycle in $\mathcal H_i$ may correspond to an even length set of cycles in $H_i$.


\section{Analysis of \RW{} when $G$ is $\epsilon$-far from bipartite}
\label{sec:main-proof}

The main proof, our analysis of \RW{} when $G$ is $\epsilon$-far from bipartite, can be kept relatively short, if we assume that the reduction from $H_i$ to $H_{i+1}$ works as supposed. Proving this reduction is the key technical challenge. We therefore first present the complete proof except for Lemma \ref{lemma:second_reduction}, whose proof appears in Section \ref{sec:proof-of-lemma:second_reduction}. Our proof follows the general outline sketched in the previous section. Our first step is to show how to obtain our starting subgraph $H_1$.


\subsection{Finding the first subgraph $H_1$}
\label{subsec:constructing-H1}

We start by taking a closer look at property (a) of Lemma~\ref{lemma:ExistenceOfH} and we show that it is implied by a simple condition on the degrees of the vertices in $H$, namely, the degree of each vertex is either $0$ or no more than a constant factor smaller than its corresponding degree in $G$. Our construction of $H_i$ in the remainder of the paper satisfies an alternate property (a'), defined below, which implies property (a).

\begin{lemma}
\label{lemma:basic}
Let $G = (V,E)$ be a graph, $t=\Theta_{\eps}(1)$ be the length of a random walk, and $H$ be a subgraph of $G$ on vertex set $V$ such that the following property holds:
\begin{itemize}
\item[(a')] for every vertex $v\in V$, either $\deg_H(v) = 0$ or $\deg_H(v) = \Omega_{\eps}(\deg_G(v))$.
\end{itemize}
Then property (a) of Lemma \ref{lemma:ExistenceOfH} is satisfied for a random walk of length $t$.
\end{lemma}

\begin{proof}
Consider a single walk $w$ in $H$ of length $t$ that finds an odd-length cycle. Since $H$ is a subgraph of $G$, the same walk exists in $G$. Furthermore, every vertex visited in $w$ must have $\deg_H(v) = \Omega_{\eps}(\deg_G(v))$. Therefore, at every step, the probability of following $w$ decreases in $G$ by at most a factor of $O_\eps(1)$, compared to $H$. Overall the probability of $w$ decreases by at most a factor of $\left(O_\eps(1)\right)^t = \left(O_\eps(1)\right)^{O_\eps(1)} =  O_\eps(1)$. Summing up over all such walks proves the lemma.
\end{proof}

We now proceed with the construction of $H_1$. We first construct a subgraph of $G$ that is a union of $\Omega_{\eps}(n)$ short edge-disjoint cycles. Then we modify it so that it satisfies property (a') and thus property (a).

We make use of the following Klein-Plotkin-Rao decomposition theorem \cite{KPR93}.

\begin{lemma}{\rm\textbf{\cite{KPR93}}}
\label{lemma:partition-into-conn-comp-of-small-diam}
Let $G = (V,E)$ be a planar graph and let $\delta$ be a parameter, $0 < \delta < 1$. There is a set of at most $\delta |V|$ edges in $G$ whose deletion decomposes $G$ into connected components, where the distance (in the original graph $G$) between any two nodes in the same component is $O(1/\delta^2)$.
\end{lemma}
Lemma \ref{lemma:partition-into-conn-comp-of-small-diam} helps us realize the first part of the plan: showing that there are many short edge-disjoint odd-length cycles in a highly non-bipartite graph.

\begin{lemma}[Many short odd-length cycles]
\label{lemma:first_reduction}
In every planar graph $G=(V,E)$ that is $\eps$-far from bipartite, there exists a collection $C$ of \ $\Omega_{\eps}(|V|)$ edge-disjoint odd-length cycles of length at most $k=k(\eps)=O(\eps^{-2})$.
\end{lemma}

\begin{proof}
We find the cycles one by one. Suppose that we have already found in $G$ a set of $\ell$ edge-disjoint odd-length cycles of length at most $k=k(\eps)$ each, where $\ell < \frac{\eps |V|}{2k}=\Omega_{\eps}(|V|)$. We show the existence of one more such cycle, which by induction yields the lemma. Let $G^\star$ be the subgraph of $G$ obtained by removing the $\ell$ edge-disjoint odd-length cycles of length at most $k$ each. Since $\ell < \frac{\eps |V|}{2k}$, $G^\star$ is obtained by removing less than $\frac{\eps |V|}{2}$ edges, and hence $G^\star$ is $\eps/2$-far from bipartite. Apply Lemma \ref{lemma:partition-into-conn-comp-of-small-diam} to $G^\star$ with $\delta = \frac{\eps}{2}$ and let $H$ be the resulting graph. Since $G^\star$ is $\eps/2$-far from bipartite, $H$ is not bipartite. Let us consider a connected component $C_H$ of $H$ that is not bipartite and let $v$ be a vertex from $C_H$. Build a BFS tree from $v$ in $G^\star$. Since $C_H$ is not bipartite, there must exist two vertices $u_1$ and $u_2$ in $C_H$ that have the same distance from $v$ and that are connected by an edge in $H$ (otherwise, we could define a bipartition of $C_H$ by the parity of the distance from $v$ in the BFS tree). Let $v'$ be the last common vertex on the paths from $v$ to $u_1$ and from $v$ to $u_2$ in the BFS tree. The cycle in $G^\star$ that starts at $v'$, goes to $u_1$ via the BFS tree edges, then takes the edge connecting $u_1$ and $u_2$, and finally returns to $v'$ via the BFS tree edges. Let $k'= O(1/\eps^2)$ be the bound on the diameter of $C_H$ (in $G^\star$) that follows from Lemma~\ref{lemma:partition-into-conn-comp-of-small-diam}.
Since the BFS tree is a shortest path tree (from $v$), this cycle has length $k = k(\eps) = 2k' + 1 = O(1/\eps^2)$.
\end{proof}

Given any set of cycles $C$ on vertex set $V$, we write $G(C)$ to denote the graph on vertex set $V$ that is induced by $C$, i.e., $G(C)=(V,E_C)$ with $E_C$ being the union of the edges of the cycles in $C$. While Lemma \ref{lemma:first_reduction} provides us with a graph $G(C)$ that has a linear number of disjoint short odd-length cycles, it is by no means clear that this new graph $G(C)$ satisfies property (a) of Lemma \ref{lemma:ExistenceOfH}. However, we show in the next lemma that there is always a subset $C' \subseteq  C$ with cardinality $| C'| = \Omega_{\eps} (|C|)$ such that the graph $G(C')$ satisfies property (a) via showing that it satisfies property (a').

\begin{lemma}[Transformation to obtain property (a')]
\label{lemma:transformation}
Let $G=(V,E)$ be a planar graph. Let $C$ be a set of $\Omega_\eps(|V|)$ edge-disjoint cycles on $V$ in $G$, each of length at most $k$, for some
$k = k(\eps) = O_{\eps}(1)$. Then there exists a subset $C' \subseteq C$ with $|C'| = \Omega_{\eps}(|V|)$ such that the graph $G(C')$ satisfies
condition (a') of Lemma \ref{lemma:basic}. That is, for every $v \in V$, either $\deg_{G(C')}(v) = 0$ or $\deg_{G{(C')}}(v) = \Omega_\eps(\deg_G(v))$.
\end{lemma}

\begin{proof}
We construct the subset $C'$ by deleting some cycles from $C$. The process of deleting the cycles is based on the comparison of the original degree of the vertices with the current degree in $G(C')$. To implement this scheme, we write $\deg_G(v)$ to denote the degree of $v$ in the original graph $G$ and we use the term \emph{current degree} of a vertex $v$ to denote its current degree in the graph $G(C')$ induced by the \emph{current} set of cycles $C'$ (where ``current'' means at a given moment in the process). Let $\alpha = |C| / |V| = \Omega_\eps(1)$. We repeat the following procedure as long as possible: if there is a non-isolated vertex $v \in V$ with current degree in $G(C')$ at most $\frac{\alpha}{12} \deg_G(v)$, then we delete from $C'$ all cycles going through $v$. To estimate the number of cycles deleted, we charge to $v$ the number of deleted cycles in each such operation. Observe that each $v \in V$ will be processed not more than once. Indeed, once $v$ has been used, it becomes isolated, and hence it is not used again. Therefore, at most $\frac{\alpha}{12} \deg_G(v)$ cycles from $C'$ can be charged to any single vertex. This, together with the inequality $\sum_{v \in V} \deg_G(v) \le 6 |V|$ by planarity of $G(C')$, implies that the total number of cycles removed from $C$  to obtain $C'$ is upper bounded by $\sum_{v \in V} \frac{\alpha}{12} \deg_G(v) \le \frac{\alpha}{2} |V|$. Since $|C| = \alpha |V|$, we conclude that $|C'| \ge |C| - \frac{\alpha}{2} |V| = \frac{\alpha}{2} |V| = \Omega_{\eps}(|V|)$.
\end{proof}

Thus, by Lemma \ref{lemma:transformation}, we can construct a subgraph $H_1:=G(C')$ of $G$ that is composed of a collection of $\Omega_{\eps}(|V|)$ odd-length edge-disjoint cycles that satisfy property (a) of Lemma~\ref{lemma:ExistenceOfH}.


\subsection{Constructing $H_{i+1}$ from $H_i$}

We begin with presenting some challenges of our construction and describe why two most natural lines of extending the analysis from Section \ref{subsec:constructing-H1} fail. After that, to facilitate our analysis, we first describe our framework in Section \ref{subsec:graph-repr} and then, in Section \ref{subsec:main-reduction}, we present details of the construction of $H_{i+1}$.


\subsubsection{On the challenges of the analysis}
\label{subsec:failures}

It is tempting to try to make a shortcut and avoid the need of constructing $H_2, \dots, H_k$ and prove directly that a random walk finds an odd-length cycle in $H_1$ by showing that a fixed cycle is found with probability $\Omega_{\eps}(1/|V|)$. Such a statement would be trivially true for graphs with a constant maximum degree, but it is false for arbitrary planar graphs, as  illustrates the example in Figure \ref{figure:key}.
\begin{figure}[t]
\centerline{\includegraphics[width=0.4\textwidth,angle=90]{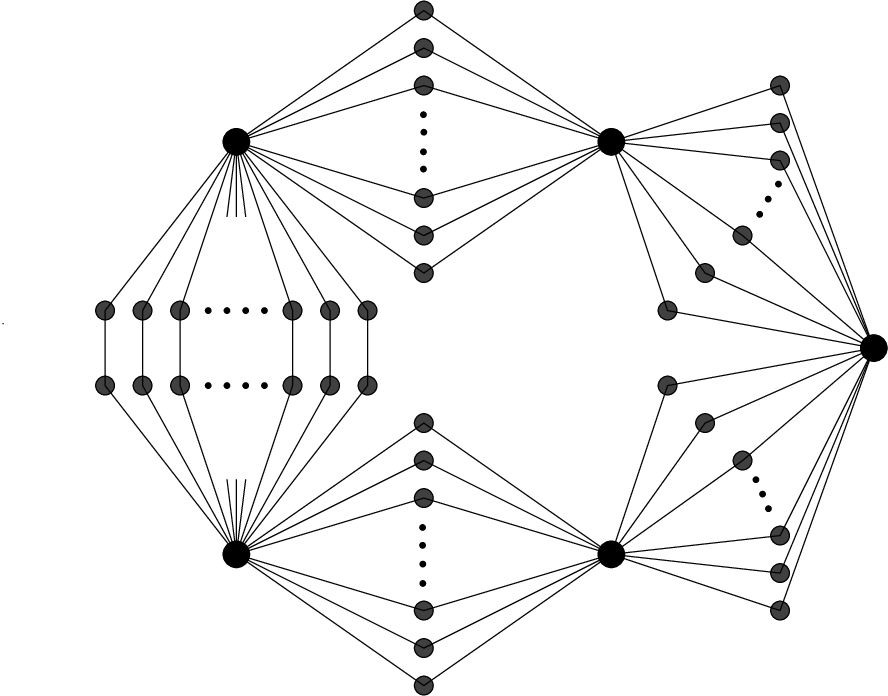}}
\caption{An illustrative example.}
\label{figure:key}
\end{figure}
The graph in this example is composed of many parallel cycles that intersect in a few vertices of high degree. It is easy to see that a random walk finds an odd-length cycle with constant probability. However, any fixed cycle is only found with sub-constant probability.\footnote{One can show that the probability of finding any fixed cycle is $n^{-\Omega_\eps(1)}$.} This implies that in our arguments it is important to \emph{exploit parallel-like structure} in the graphs.

Similarly, it could be tempting to hope that algorithm \RW{} finds with constant probability an odd-length cycle in \emph{an arbitrary graph} $G$ (not necessarily planar) that is a union of a linear number of edge-disjoint odd-length cycles of constant length.
This, however, turns out not to be the case. Consider
an expander with girth $\omega(1)$, in which every vertex has degree $\omega(1)$ (see, e.g., \cite{LPS88}). Then replace every edge by a cycle of length 3. It is easy to see that a random walk of constant length will find an odd-length cycle only if it finds one of the new cycles of length~3. However, since the probability of this event is $o(1)$, the algorithm will not find a cycle with constant probability. This implies that in our arguments it is crucial to \emph{exploit planarity}.

Note also that a constant upper bound on the graph diameter is not sufficient. To see this, consider a clique on $\sqrt{n}$ vertices and, similarly as above, replace every edge by a cycle of length $3$. The resulting graph has $\Theta(n)$ vertices. A necessary condition for a random walk to discover a cycle is to discover both edges that belong to the same cycle from the starting vertex or it has to return to a previously visited vertex of the clique from another clique vertex. But this is unlikely with a constant number of random steps since the vertices of the clique have non-constant degree.


\subsubsection{Graph representation}
\label{subsec:graph-repr}

All our graphs $H_i$ are formed by sets $C_i$ of short odd-length edge-disjoint cycles from $G=(V,E)$. We also have that $C_1 \supseteq C_2 \supseteq \dots \supseteq C_k$ and $|C_k| = \Omega_{\eps}(|V|)$, and thus $H_1 \supseteq H_2 \supseteq \dots \supseteq H_k$. For a set of cycles $C$ on vertex set $V$, we write $G(C)$ to denote the graph with vertex set $V$ and edge set being the set of edges from the cycles in $C$. Thus, in our reduction $H_i = G(C_i)$. Recall that in addition to the graphs $H_i$ we use the multigraphs $\mathcal H_i$. In order to define $\mathcal H_i$, we first require a partition $P_i:V \rightarrow V$ of the vertices that describes how the edges are contracted. The idea is that the vertex set $P^{-1}(u)$ is contracted into vertex $u$. In order to be meaningful, this partition has to satisfy the properties given in the following definition.

\begin{definition}
\label{def:good-partition}
Let $C$ be an arbitrary set of disjoint cycles in $G$. A partition $P:V \rightarrow V$ is called \emph{good} for $C$ if it satisfies the following four properties.
\begin{itemize}
\item If $u \in V$ is in the image of $P$, then $P(u)=u$.
\item If $u\in V$ is not contained in any cycle from $C$, then $P(u)=u$.
\item For each cycle $\mathfrak{c} \in C$ and each partition class $P^{-1}(u)$, where $u$ belongs to the image of $P$, if $P^{-1}(u)$ contains a vertex of $\mathfrak{c}$, then $\mathfrak{c}$ also contains $u$.
\item For each cycle $\mathfrak{c} \in C$ and each partition class $P^{-1}(u)$, where $u$ belongs to the image of $P$, one of the following is true:
    \begin{inparaenum}[\it (i)]
    \item $P^{-1}(u)$ contains all vertices of $\mathfrak{c}$, or
    \item $P^{-1}(u)$ contains no vertex of $\mathfrak{c}$, or
    \item $P^{-1}(u)$ induces a path in $\mathfrak{c}$.
    \end{inparaenum}
\end{itemize}
A vertex $u$ in the image of $P$ is called \emph{the head of the partition class $P^{-1}(u)$}.	
\end{definition}

\begin{figure}[t]
\centerline{\includegraphics[width=.9\textwidth]{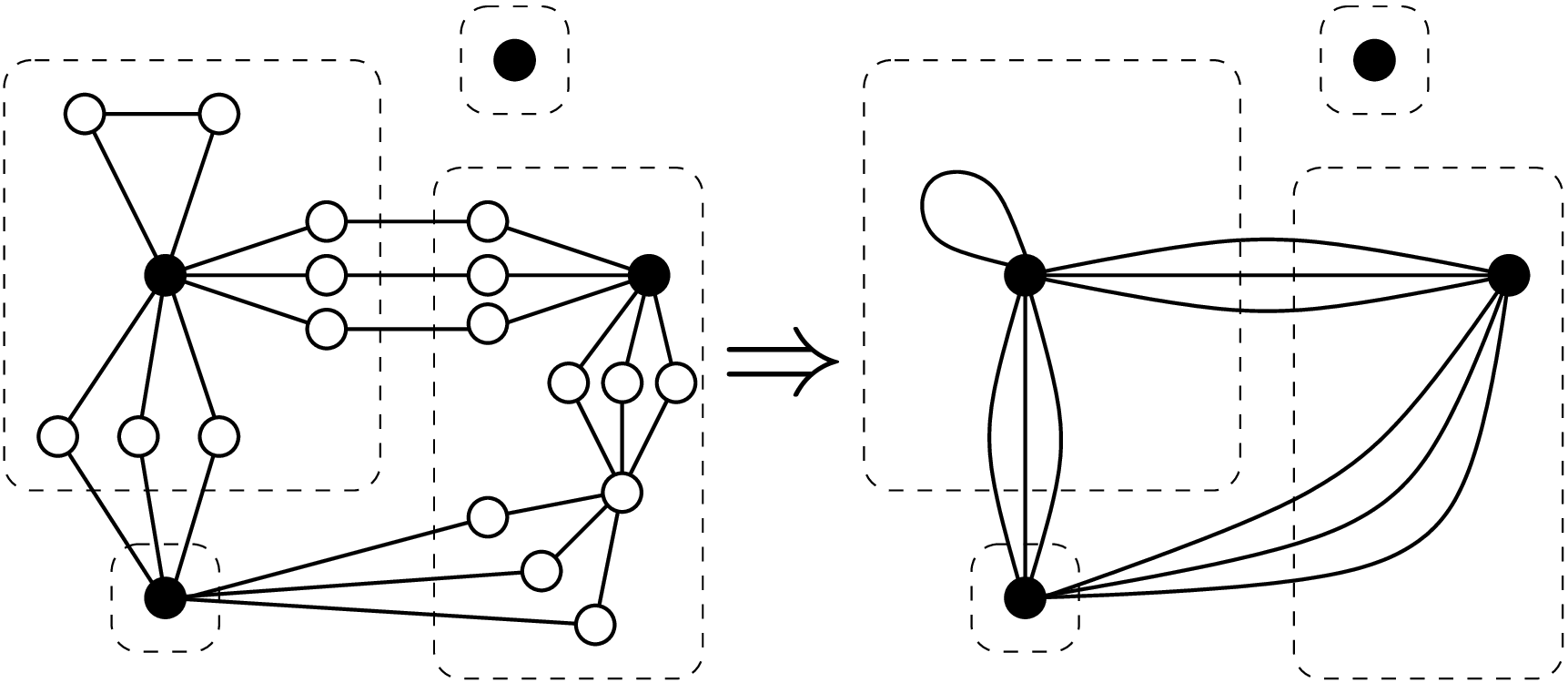}}
\caption{A sample good partition of vertices and the resulting contracted graph. The solid circles denote heads of partition classes.}
\label{fig:partition}
\end{figure}

An example of a good partition is presented in Figure \ref{fig:partition} (left side). The first property ensures that all partition classes have a proper head, i.e., a vertex into which all other vertices are contracted. The second property ensures that each isolated vertex in $G(C)$ has its own partition class. The third and fourth properties ensure that we can apply the contractions to each cycle from $C$ individually. This way we may also obtain from $C$ a multiset of contracted cycles $\mathcal C$.

Notice that the above definition implies that any partition class $P^{-1}(u)$, $u$ in the image of $P$, forms a connected subgraph in $G(C)$. Indeed, assume that $v$ and $w$ are in $P^{-1}(u)$, $v \ne w$. By the second property, any of $v$ and $w$ which is not equal to $u$ is contained in a cycle. By the third property, either $v=u$ or $v$ is contained in a cycle, and hence $v$ and $u$ are in the same cycle; similarly, either $w=u$ or $w$ and $u$ are in the same cycle. Therefore, there is a path connecting $v$ and $w$.

In our analysis, we modify a given partition $P$ in two ways: by \emph{edge contractions} and by \emph{cycle removals}. As a result of contracting an edge $(u,v)$, the partition classes of $u$ and $v$ are merged. Deleting one (or more) cycles from $C$ may create an isolated vertex (with degree $0$) in $G(C)$. In this case, the second property in Definition \ref{def:good-partition} may be violated. Therefore, whenever we create an isolated vertex $u$ in $G(C)$ by deleting a cycle from $C$, we define $P(u) = u$ to satisfy the second condition. It is easy to verify that the other conditions are still satisfied after this modification has been applied to all newly created isolated vertices.

We abuse the notation and write $P(\mathfrak{c})$ to denote the cycle obtained from $\mathfrak{c}$ by contracting its vertices according to $P$. We remark that the contracted cycles may be degenerated to self-loops or cycles of the form $(u,v,u)$. For example, if all vertices of $\mathfrak{c}$ are contained in the same partition class of $P$, then the resulting cycle is a self-loop incident to the head of the partition class. For any set of cycles $C$, we also write $P(C)$ to denote the union of $P(\mathfrak{c})$ for all $\mathfrak{c} \in C$. Finally, we define $\mathcal G_{P}(C)$ to be the multigraph whose vertex set is the image of $P$ and whose edge multiset is the multiset of edges of cycles in $P(C)$. The right side of Figure \ref{fig:partition} contains the resulting multigraph for the input set of cycles and partition on the right.
For a multigraph $\mathcal G$, we write $\deg_{\mathcal G}(u)$ to denote the number of edges in $\mathcal G$ that are incident to $u$ (i.e., self-loops contribute $1$ to the degree).

In our main reduction, we define $\mathcal H_i = \mathcal G_{P_i}(C_i)$. The multigraph $\mathcal H_i$ has the following interpretation. An edge $(u,v)$ in $\mathcal H_i$ represents a path of $H_i$ of length $O_\eps(1)$. Therefore, with probability $\Theta_{\eps}(1/\deg_{\mathcal H_i}(u))$, a random walk on $G(C_i)$ of (suitable) constant length starting in $u$ reaches $v$. If there are $q$ edges $(u,v)$ in $\mathcal H_i$ then their union represents a subgraph of $H_i$ that for random walks behaves like a set of $q$ parallel non-intersecting paths. In particular, the probability of moving from $u$ to $v$ is $\Theta_{\eps}(q/\deg_{\mathcal H_i}(u))$, where $q$ is the number of edges $(u,v)$ in $\mathcal H_i$.

Furthermore, a random walk starting from a vertex in $P^{-1}(u)$ reaches  vertex $u$ with constant probability after a constant number of steps.

With this definition, we can extend our notion of random walks to multigraphs in the following natural way (where we define the notion of odd-parity cycles below).

\medskip
\begin{algo}
\RW\,${(\mathcal{G}= (U,\mathcal{E}),t)}$:
\begin{itemize}
\item Pick a random vertex $u \in U$, such that any $u \in U$ is chosen with probability $\frac{|P^{-1}(u)|}{|V|}$.
\item  Perform a random walk of length $t$ from $v$, where the probability to move from a vertex $u \in U$ to a vertex $w\in U$ is $|\{(u,w)\in \mathcal{E}\}| / \deg_{\mathcal{G}}(u)$.
\item  If the random walk found an odd-parity cycle, then \textbf{reject}.
\item  If not, then \textbf{accept}.
\end{itemize}
\end{algo}

\paragraph{Parities of edges and lengths of cycles.}

Since the contraction of edges influences the length of cycles and since we are searching for cycles of odd-length in $G(C)$, we need a way to keep track of the parity of the lengths of the cycles and paths explored and contracted. It would be easy, if we were just dealing with fixed cycles, but since we would like to allow combinations of cycles, we need to encode the contractions in the graph in a way that allows us to use combinations. This can be done as follows. For every cycle $\mathfrak{c} \in C_i$, we have a corresponding contracted cycle in $\mathfrak{c}' \in P_i(C_i)$ on vertex set $U$. Each edge in $\mathfrak{c}'$ corresponds to a path in $\mathfrak{c}$ and we define the parity of an edge $(u,v)$ in $\mathfrak{c}'$ as $0$, if the length of the path from $u$ to $v$ in $\mathfrak{c}$ is \emph{even} and $1$, if the length is \emph{odd}. In the case that $u$ and $v$ are the only two vertices in $\mathfrak{c}'$, there are two edges connecting $u$ and $v$ and one of them has parity $0$ while the other one has parity $1$. This way, a cycle in the multigraph $\mathcal{G}$ has odd parity if it contains an odd number of edges with parity $1$.


\subsubsection{Main reduction from $H_i$ and $\mathcal H_i$ to $H_{i+1}$ and $\mathcal H_{i+1}$}
\label{subsec:main-reduction}

The next lemma states our main technical contribution: the main reduction with its properties. We apply it several times to reduce the lengths of our cycles to $1$. For the clarity of presentation, we postpone the proof of this lemma to Section \ref{sec:proof-of-lemma:second_reduction}.

\begin{lemma}
\label{lemma:second_reduction}
Let $t = O_{\eps}(1)$, let $G = (V,E)$ be a planar graph, and let $\ell \in \NAT$, $\ell \ge 2$. Let $C$ be a set of $\Omega_{\eps}(|V|)$ odd-length cycles in $G$. Let $P$ be a partition that is good for $C$ such that all cycles in $P(C)$ have length at most $\ell$. Then we can construct a set of cycles $C^* \subseteq C$ with $|C^*| = \Omega_{\eps}(|C|)$, and a partition $P^*$ that is good for $C^*$, such that the following properties are satisfied:
\begin{itemize}
\item every cycle in $P^*(C^*)$ has length at most $\ell-1$, and
\item if the probability that \RW$(\mathcal{G}_{P^*}(C^*),t)$ finds an odd-parity cycle is $\Omega_{\eps}(1)$, then also the probability that \RW$(\mathcal{G}_{P}(C),3t)$ finds an odd-parity cycle is $\Omega_{\eps}(1)$.
\end{itemize}
\end{lemma}

This lemma is used to construct $H_{i+1}$ and $\mathcal H_{i+1}$. We take as the input $C = C_i$ and $P = P_i$, and apply Lemma \ref{lemma:second_reduction} to obtain $C_{i+1} = C^*$ and $P_{i+1} = P^*$, giving $H_{i+1} = G(C^*)$ and $\mathcal{H}_{i+1} = \mathcal{G}_{P^*}(C^*)$.

\junk{
\begin{lemma}
Let $t=O_{\eps}(1)$. Let $G = (V,E)$ be a planar graph and let $\ell \in \NAT$, $\ell \ge 2$. Let $C_i$ be a set of $\Omega_{\eps}(|V|)$ edge disjoint odd-length cycles in $G$. Let $P_i$ be a partition that is good for $C_i$ such that all cycles in $P_i(C_i)$ have length at most $\ell$. Then we can construct a set of cycles $C_{i+1} \subseteq C_i$ with $|C_{i+1}| = \Omega_{\eps}(|C_i|)$, and a partition $P_{i+1}$ that is good for $C_{i+1}$, such that the following properties are satisfied:
\begin{itemize}
\item every cycle in $P_{i+1}(C_{i+1})$ has length at most $\ell-1$, and
\item if the probability that \RW$(\mathcal H_{i+1},t)$ finds an odd-parity cycle is $\Omega_{\eps}(1)$, then also the probability that \RW$(\mathcal H_i,2t)$ finds an odd-parity cycle is $\Omega_{\eps}(1)$.
\end{itemize}
\end{lemma}
}


\subsubsection{Useful property of $\mathcal H_k$}
\label{subsec:final-self-loops}

The first property of the construction in Lemma \ref{lemma:second_reduction} implies that the length of the cycles in $P_i(C_i)$ decreases with increasing $i$. We therefore apply Lemma \ref{lemma:second_reduction} $k-1$ times, where $k$ is the original upper bound for the cycle length in Lemma \ref{lemma:transformation}. As a result, we obtain a sequence $C_1 \supseteq \dots \supseteq C_k$ consisting of sets of cycles such that $P_k(C_k)$ consists solely of self-loops. We use the following property of graphs induced by such self-loops.

\begin{lemma}
\label{lemma:Selfloops}
Let $G=(V,E)$ be a planar graph. Let $C_k$ be a set of edge disjoint odd-length cycles in $G$ and $P_k$ be a partition that is good for $C_k$. If all cycles in $P_k(C_k)$ are self-loops and $|C_k| = \Omega_{\eps}(|V|)$, then the probability that a $1$-step random walk finds an odd-parity cycle in $\mathcal H_k$ is $\Omega_{\eps}(1)$.
\end{lemma}

\begin{proof}
Since $C_k$ is a collection of $\Omega_{\eps}(|V|)$ edge disjoint cycles in $G$, the underlying graph $H_k$ has $\Omega_{\eps}(|V|)$ edges. Therefore since $H_k$ is planar, $H_k$ must have $\Omega_{\eps}(|V|)$ non-isolated vertices. Each of these non-isolated vertices is contracted in $P_k$ into a vertex that is incident to a self-loop. Therefore, the probability to sample a vertex incident to a self-loop is $\Omega_{\eps}(1)$ and then in one step we find a self-loop. By the definition of parities, the detected self-loop  has odd parity.
\end{proof}


\subsection{Proof of Lemma \ref{lemma:ExistenceOfH} and Theorem \ref{thm:main-bipartiteness}: analyzing \RBE}

Now we are ready to prove Lemma \ref{lemma:ExistenceOfH} and with this, our main theorem, Theorem \ref{thm:main-bipartiteness}.

In order to prove Lemma \ref{lemma:ExistenceOfH}, we show the existence of $H$ as required in Lemma \ref{lemma:ExistenceOfH}, by constructing a sequence of subgraphs $G \supseteq H_1 \supseteq \dots \supseteq H_k$ for some $k=k(\epsilon)$, such that each $H_i$, $1 \le i \le k$, satisfies property (a) from Lemma \ref{lemma:ExistenceOfH} and such that $H:=H_1$ also satisfies property (b) from Lemma \ref{lemma:ExistenceOfH}.

We know by Lemma \ref{lemma:first_reduction} that $G$ contains a set $C$ of $\Omega_{\eps}(|V|)$ odd-length cycles of length at most $k = O_{\eps}(1)$. Furthermore, by Lemma \ref{lemma:transformation}, there is a set $C_1 \subseteq C$ of $\Omega_{\eps}(|V|)$ odd-length cycles of length at most $k$ such that if \RW$(G(C_1),t)$ on $G(C_1)$ finds an odd-length cycle with probability $\Omega_{\eps}(1)$, then so does \RW$(G,t)$.
We then define $H_1 = G(C_1)$ and $P_1:V \rightarrow V$ to be the identity. Clearly, $P_1$ is good for $C_1$. Then we apply $k-1$ times Lemma \ref{lemma:second_reduction} to obtain sets of cycles $C_i$ and partitions $P_i$ that satisfy the properties of Lemma \ref{lemma:second_reduction}. In particular, since $k = O_\eps(1)$, we know that $\mathcal H_k$ contains $\Omega_{\eps}(|V|)$ self-loops. Moreover, if \RW$(\mathcal H_k, T)$ finds an odd-parity cycle with probability $\Omega_{\eps}(1)$, then so does \RW$(\mathcal H_1, T \cdot 3^{k-1})$, which in turn implies that \RW$(G, T \cdot 3^{k-1})$ also finds an odd-length cycle with probability $\Omega_{\eps}(1)$. Thus, we only need to prove that \RW$(\mathcal H_k, T)$ with $T = O_{\eps}(1)$ finds an odd-parity cycle with probability $\Omega_{\eps}(1)$. This follows immediately from Lemma \ref{lemma:Selfloops}. This concludes the proof of Lemma \ref{lemma:ExistenceOfH} by setting $H = H_1 = G(C_1)$ since 
a random walk in $\mathcal H_1$ behaves identically to a random walk in $H_1$.

Once we have Lemma \ref{lemma:ExistenceOfH}, Theorem \ref{thm:main-bipartiteness} follows immediately. Indeed, we already observed that to prove Theorem \ref{thm:main-bipartiteness} it suffices to show that for any planar graph $G = (V,E)$ that is $\eps$-far from bipartite, \RW$(G,t)$ finds an odd-length cycle with probability $\Omega_{\eps}(1)$ for $t = O_{\eps}(1)$, which follows from Lemma~\ref{lemma:ExistenceOfH}.


\section{Proof of Lemma \ref{lemma:second_reduction}}
\label{sec:proof-of-lemma:second_reduction}

To complete the analysis, it remains to prove Lemma \ref{lemma:second_reduction}. We start with an overview of the proof. The main idea is to thin out the current set of cycles $C$ to ensure that we can define a set of ``contractions'' to decrease the length of each of the remaining cycle. Further care is needed to ensure that after performing the contractions, we still maintain a good partition for the set of remaining cycles. This means that we are not allowed to contract edges that ``shortcut'' other remaining cycles. In order to avoid this we ensure that one of the contracted vertices, say $v$, either has only one distinct neighbor in $\mathcal H_i$ (in which case all cycles that involve this vertex must be of the form $(u,v,u)$) or it has exactly two distinct neighbors, say $x$ and $y$, and all cycles that involve $v$ contain the edges $(x,v)$ and $(v,y)$. Furthermore, in the latter case, we ensure that all edges $(x,v)$ have the same parity and all edges $(v,y)$ have the same parity, as otherwise the parity after the contraction would not be well-defined. This is captured by the following definition.

\begin{definition}
Let $C$ be a set of cycles on vertex set $V$ and $P$ be a partition that is good for $C$. We say a vertex $v \in P(V)$ is \emph{well-contractible} in $\mathcal G_P(C)$ if it satisfies one of the following conditions:
\begin{itemize}
\item $v$ has only one distinct neighbor in $\mathcal G_P(C)$, or
\item $v$ has two distinct neighbors $x$ and $y$ in $\mathcal G_P(C)$, and
    \begin{itemize}[$\diamond$]
    \item all cycles in $C$ that contain $v$ contain also both $x$ and $y$, and
    \item all edges $(x,v)$ have the same parity and all edges $(v,y)$ have the same parity.
    \end{itemize}
\end{itemize}
\end{definition}

Using this definition, our goal is to thin out the cycles such that every cycle has a well-contractible vertex. We also need to ensure that our contractions using well-contractible vertices do not interfere with each other. Therefore we additionally require that the set of well-contractible vertices forms an independent set.

After finding a set of cycles $C'$ (and good partition $P'$) with the above properties, we still have to do a  ``clean-up'' phase
to ensure that if a random walk finds with constant probability a cycle in $\mathcal G_P'(C')$, then this is also true in $\mathcal G_P(C)$. In order to achieve this, we use Lemma \ref{lemma:transformation} to ensure that the degree of every non-isolated vertex in $\mathcal G_P'(C')$ is at least a constant fraction of its degree in $\mathcal G_P(C)$.


\paragraph{Outline.}

The roadmap is now as follows. We first find a subset of cycles such that every cycle has a well-contractible vertex. In order to do so, we first construct a large subset of cycles such that every cycle has a vertex with at most {6} distinct neighbors (Lemma \ref{lemma:small-vertices}). Then we argue that such a set of cycles satisfies at least one of the following two conditions: (1) It already has many self-loops, in which case we can just take this subset of self-loops and we are done with our reduction, or (2) we can remove all self-loops and process the remaining (linear number of) cycles to ensure that every cycle contains a well-contractible vertex (Lemma~\ref{lemma:contractible-vertices}).

\begin{lemma}
\label{lemma:small-vertices}
Let $C$ be a set of edge-disjoint cycles on vertex set $V$ that are of length at most $k$ and such that $G(C)$ is a planar graph. Let $P$ be a partition that is good for $C$. Then there is a set $C' \subseteq C$ of size at least $\frac{1}{4k+2} |C|$ such that every cycle in $P(C')$ has a vertex with at most 6 distinct neighbors in $\mathcal G_{P}(C')$.
\footnote{We slightly abuse notation here as $P$ may be no longer a good partition due to vertex removal. In this case, we can still define $\mathcal G_{P}(C')$ the same way as above.}
\end{lemma}

\begin{proof}
We prove the lemma by presenting an algorithm that takes as its input a set of cycles $C$ with planar $G(C)$ and a partition $P$ that is good for $C$. The algorithm computes a subset $C'$ that satisfies the properties of the lemma. The algorithm consists of two phases.

In the first phase, we partition $C$ into \emph{levels}, iteratively removing the cycles until $C$ is empty. We start with $C$ being the input set of cycles. In the $j^{\text{th}}$ iteration, we choose an arbitrary vertex $u_j$ with at most 6 distinct neighbors in $\mathcal G_P(C)$ and at least one incident cycle. Such a vertex exists since $G(C)$ is a planar graph, and so by Euler's formula there exists a vertex in
{$\mathcal G_P(C)$}
with at most 5 distinct neighbors other than itself. Taking into account that
{$\mathcal G_P(C)$}
may also contain self-loops, we find a vertex with at most 6 distinct neighbors. Here $C$ refers to the current set $C$, i.e., after the removal of the sets from the previous iterations of the repeat-loop. Every cycle from $C$ that contains $u_j$ is removed from $C$. If a cycle $\mathfrak{c}$ is removed in the $j^{\text{th}}$ iteration, its \emph{level} $\ell(\mathfrak{c})$ is defined to be $j$.

In the second phase, we start again with $C$ being the input set of cycles. We now iterate through the levels in decreasing order. For each level $j$, we let $A(j)$ denote the current subset of cycles in $C$ at level $j$. Note that by definition of the level, all cycles in $A(j)$ must contain vertex $u_j$. Furthermore, we define $B(j)$ to be the subset of cycles of $C$ that contain $u_j$ and have a level smaller than $j$. We observe that if we remove all cycles in $B(j)$ from $C$, then every cycle in $A(j)$ contains a vertex (e.g., vertex $u_j$) with at most 6 distinct neighbors in $\mathcal G_P(C)$. The second phase explores this observation. For each $j$, we compare the size of $A(j)$ to the size of $B(j)$. If $A(j)$ is sufficiently large, i.e., at least a $\frac{1}{2k}$ fraction of $B(j)$, then we keep $A(j)$ and remove $B(j)$ from $C$; otherwise, we remove $A(j)$. Below we argue that at most half of the cycles from $C$ are removed because they are in some removed set $A(j)$. Furthermore, for every $2k$ cycles that are removed because they are contained in a set $B(j)$, at least one cycle remains in $C$. This allows us to deduce the lemma. Details follow after the pseudocode describing more formally the process.

\medskip
\begin{algo}
\AL\,(set $C$ of cycles and a partition $P$ that is good for $C$)
\begin{itemize}
\item $j=1$
\item $C''=C$
\item[] \textbf{\!\!Phase 1:}
\item Repeat until $C$ is empty:
    \begin{itemize}[$\circ$]
    \item Let $u_j$ be a non-isolated vertex that has at most 6 distinct neighbors in $\mathcal G_{P}(C)$
    \item For all cycles $\mathfrak{c} \in C$ that contain $u_j$, let $\ell(\mathfrak{c})=j$
		\item Remove from $C$ all cycles that contain $u_j$
    \item $j=j+1$
		\end{itemize}
\item[]\textbf{\!\!Phase 2:}
\item $C=C''$
\item Repeat until $j=1$:
     \begin{itemize}[$\circ$]
		 \item $j=j-1$
		 \item $A(j) = \{\mathfrak{c} \in C: \ell(\mathfrak{c}) = j\}$
		 \item $B(j) = \{\mathfrak{c} \in C: \ell(\mathfrak{c}) < j \text{ and } \mathfrak{c} \text{ contains } u_j\}$
	   \item if $|A(j)| \ge \frac{1}{2k} \cdot |B(j)|$ then
		       $C = C \setminus B(j)$ else $C=C \setminus A(j)$
		 \end{itemize}
\item Return $C$
\end{itemize}
\end{algo}

It remains to prove the correctness of the algorithm. We first observe that Phase 1 terminates since $\mathcal G_{P}(C)$ is a planar graph and therefore by Euler's formula, it has a vertex with at most
{6}
neighbors (this also holds during the execution of the algorithm since planarity is closed under edge removal and contractions).

It remains to analyze Phase 2 of the algorithm. Every cycle in $C$
\begin{inparaenum}[(a)]
\item is removed because it is contained in some set $A(j)$ that is removed from $C$ in Phase 2, or
\item is removed because it is contained in some set $B(j)$ that is removed from $C$ in Phase 2, or
\item is not removed and stays in the final set $C$.
\end{inparaenum}
Let $a$, $b$, and $c$ be the respective numbers of cycles. Clearly, $|C|=a+b+c$ and to prove the lemma we have to show that $c \ge |C|/k$. We proceed in two steps. We first prove that $a \le \frac12 |C|$, which implies that $b + c \ge \frac12 |C|$. Then we argue that $2 k c \ge b$. This yields $(2k+1) c \ge \frac12 |C|$ and hence $c \ge \frac{1}{4k+2} \cdot |C|$.

\begin{quote}

\begin{claim}
\label{claim:a}
$a \le \frac12 |C|$.
\end{claim}

\begin{proof}
We charge the vertices from the removed sets $A(j)$ to the sets $B(j)$ and derive a bound on the sum of sizes of the sets $B(j)$. Recall that every cycle contains at most $k$ vertices. In every cycle, one vertex is the vertex that has degree at most
{6 (in $\mathcal G_P(C)$)}
when the cycle is removed in Phase 1 of the algorithm. Thus, every cycle is contained in at most $k-1$ different sets $B(j)$. It follows that
\begin{displaymath}
    \sum_{j} |B(j)| \le (k-1) \cdot |C|.
\end{displaymath}
Let $R$ denote the set of indices $j$ such that $A(j)$ is removed from $C$ during Phase 2. Observe that whenever we remove a set $A(j)$, we have $|A(j)| < \frac{1}{2k} |B(j)|$ by the condition in the process. It follows that
\begin{displaymath}
    a =
    \sum_{j \in R} |A(j)| <
    \sum_{j \in R} \frac{1}{2k} |B(j)| \le
    \frac{k-1}{2k} \cdot |C| <
    \frac{1}{2} \cdot |C|.
\end{displaymath}
\end{proof}

\begin{claim}
\label{claim:b}
$2 k c \ge b$.
\end{claim}

\begin{proof}
For every set $B(j)$ removed from $C$, we know that $|A(j)| \ge \frac{1}{2k} |B(j)|$. At the point of time when $B(j)$ is removed from $C$, the set $A(j)$ remains in $C$ because $A(j)$ and $B(j)$ are disjoint. Since we are iterating downwards through the levels of the cycles, the set $A(j)$ is also disjoint from all sets $B(j')$, $j'<j$, and so it is not removed also in any future iteration of the repeat loop. Thus, in this case each cycle from $A(j)$ remains in $C$ until the end of the process and contributes to the value of $c$. Let $R'$ be the set of indices $j$ such that $A(j)$ remains in $C$ during Phase 2 (and hence $B(j)$ is removed from $C$). Since each cycle of $A(j)$, $j\in R'$, contributes to $c$ and since sets $A(j)$ are disjoint, we obtain $\sum_{j\in R'} |A(j)| \le c$. Hence,
\begin{displaymath}
    \frac{1}{2k} b =
    \frac{1}{2k} \sum_{j\in R'} |B(j)| \le
    \sum_{j\in R'} |A(j)| \le
    c,
\end{displaymath}
which implies the claim.
\end{proof}
\end{quote}

This finishes the proof of Lemma \ref{lemma:small-vertices}, which follows from Claims \ref{claim:a} and \ref{claim:b} as argued above.
\end{proof}

Our next lemma shows that if there are no self-loops, then we can cover a large number of cycles by well-contractible vertices.

\begin{lemma}
\label{lemma:contractible-vertices}
Let $C$ be a set of edge disjoint cycles on a vertex set $V$ of length at most $k$. Let $P$ be a partition that is good for $C$ such that $\mathcal G_P(C)$ contains no self-loops. Let $Q$ be the set of vertices in $V$ that have at most 6 distinct neighbors in $\mathcal G_P(C)$. If every cycle in $C$ contains at least one vertex from $Q$, then there is a subset $C' \subseteq C$, $|C'| \ge 12^{-2k} |C|$, such that every cycle $\mathfrak{c} \in C'$ has a vertex $v\in Q$ that is a well-contractible vertex in $\mathcal G_P(C')$.\footnote{We slightly abuse notation here as $P$ may be no longer a good partition due to vertex removal. In this case, we can still define $\mathcal G_{P}(C')$ the same way as above.}
\end{lemma}

\begin{proof}
For each vertex $v \in Q$, we select $x$ and $y$ independently uniformly at random with among its neighbors in $\mathcal G_P(C)$. If $x = y$, we delete from $C$ all cycles $\mathfrak{c}$ that contain $v$ and for which $P(\mathfrak{c})$ contains any vertex other than $v$ or $x$.
If $x \ne y$, we also select independently uniformly at random parities $p_x, p_y \in \{\mbox{odd},\mbox{even}\}$. In this case, we remove from $C$ every cycle $\mathfrak{c}$ that contains $v$ unless $P(\mathfrak{c})$ contains both the edge $(v,x)$ of parity $p_x$ and the edge $(v,y)$ of parity $p_y$. Let $C'$ be the set of remaining cycles. We observe that every cycle $\mathfrak{c}$ in $C'$ contains a vertex $v\in Q$ that is well-contractible (in fact, every vertex from $Q$ contained in $\mathfrak{c}$ is well-contractible). The probability that a fixed cycle $\mathfrak{c}$ is not deleted is at least $12^{-2k}$. Thus, the expected size of $C'$ is at least $12^{-2k} |C|$, and therefore, there exists a set $C'$ of that size that satisfies the lemma.
\end{proof}

We now proceed to the main technical lemma that prepares our reduction step. Given a set $C$ of cycles and a partition $P$, we show that there is a set $Q$ and a subset $C'$ of $C$ of size $\Omega_{\epsilon}(|C|)$ such that we can simultaneously contract edges at all vertices from $Q$ to shorten all cycles in $P(C')$. Furthermore, the degree of each non-isolated vertex in the multigraph $\mathcal G_P(C')$ is comparable to its degree in $\mathcal G_P(C)$.

\begin{lemma}
\label{lemma:MainStep}
Let $\eps \in (0,1)$. Let $C$ be a set of edge-disjoint cycles on vertex set $V$ of length at most $k=O_{\eps}(1)$ such that $G(C)$ is planar and let $P$ be a partition that is good for $C$. There exists a set $C' \subseteq C$ of size $\Omega_{\epsilon}(|C|)$, a set of vertices $Q \subseteq P(V)$,
and a partition $P'$ that is good for $C'$, such that
\begin{itemize}
\item $P'$ is obtained from $P$ by setting for every vertex $u$,
    $P'(u) = \begin{cases}
        P(u) & \mbox{if $\deg_{G(C')}(u) > 0$,} \\
        u & \mbox{otherwise,}
    \end{cases}$
\item $Q$ is an independent set in $\mathcal G_{P'}(C')$,
\item every vertex in $Q$ is well-contractible in $\mathcal G_{P'}(C')$,
\item every cycle in $C'$ that is not a self-loop contains a vertex from $Q$, and
\item every non-isolated vertex $v$ from $\mathcal G_{P'}(C')$ satisfies $\deg_{\mathcal G_{P'}(C')}(v) = \Omega_{\eps}(\deg_{\mathcal G_{P}(C)}(v))$.
\end{itemize}
\end{lemma}

\begin{proof}
Let $C_1 \subseteq C$ be the subset of cycles from $C$ that are self-loops in $P(C)$ and let $C_2 = C \setminus C_1$. If $|C_1| \ge |C_2|$, then we choose $C'=C_1$, $Q = \emptyset$ and $P'$ as in the lemma to satisfy the properties of Lemma~\ref{lemma:MainStep}.

It remains to consider the case that $|C_2| > |C_1|$, in which case we select $C' \subseteq C_2$. We first set $P' = P$ and then iteratively modify $P'$ such that it remains a good partition for the current set of cycles and can be obtained as the statement of the lemma specifies.

We modify $P'$ in the way described earlier in the case of deletions of cycles. First, we set $P'(u) = u$ for every vertex $u$ with degree $0$ in $G(C_2)$. Then we apply Lemma \ref{lemma:small-vertices} with $C_2$ and $P'$ to obtain a set of cycles $C_2'$ of size $\Omega_{\epsilon}(|C|)$, such that every cycle from $P'(C_2')$ contains at least one vertex with at most 6 distinct neighbors. In order to maintain the first property in the lemma statement, we modify $P'$ by setting $P'(v) = v$ for new isolated vertices in $G(C_2')$. Note that after the modification every cycle from $P'(C_2')$ still contains at least one vertex with at most 6 distinct neighbors.

Next, we apply Lemma \ref{lemma:contractible-vertices} to $C_2'$ and the current $P'$. We obtain sets $Q'$ and $C_2'' \subseteq C_2'$. $Q'$ is the set of vertices that have at most $6$ distinct neighbors in $\mathcal G_{P'}(C_2'')$. $|C_2''| = \Omega_{\epsilon}(|C_2'|) = \Omega_{\epsilon}(|C|)$ and every cycle in $C_2''$ contains a vertex from $Q'$ that is well-contractible. Then we apply Lemma \ref{lemma:transformation} to $C_2''$ to obtain the final set of cycles $C' = \Omega_{\epsilon}(|C''_2|)$. We modify $P'$ as before, by setting $P'(v)$ for new isolated vertices in $\mathcal G_{P'}(C')$. Clearly, $C'$ and $P'$ satisfy the first and the fifth property specified in the lemma.

We set $Q$ to be any maximal independent subset of vertices in $Q'$ that are well-contractible in $\mathcal G_{P'}(C')$. This ensures both the second and third property. It remains to show that the fourth property is satisfied. By the construction of $C'$, every cycle $\mathfrak{c} \in C'$ contains a vertex $v \in Q'$ that is well-contractible in $\mathcal G_{P'}(C')$. If $v \in Q$, we are done. Otherwise, since $Q$ is maximal, we know that a neighbor of $v$ in $\mathcal G_{P'}(C')$ is in $Q$. By the definition of well-contractible vertices, this neighbor is also in $\mathfrak{c}$, which completes the proof of the lemma.
\end{proof}


\subsection{Finalizing the proof of Lemma \ref{lemma:second_reduction}}

We are ready to complete the proof of Lemma \ref{lemma:second_reduction}.
For a planar graph $G = (V,E)$, let $C$ be a set of $\Omega_{\eps}(|V|)$ odd-length cycles in $G$ and let $P$ be a partition that is good for $C$ such that all cycles in $P(C)$ have length at most $\ell\ge 2$. We first apply Lemma \ref{lemma:MainStep} to obtain the sets $C'$ and $Q$ and the partition $P'$. We set $C^* = C'$.

Now we consider all vertices in $Q$ and contract each of them into one of its neighbors. This construction is well-defined since $Q$ is an independent set by the second property of Lemma \ref{lemma:MainStep}. For each vertex $u \in Q$, we select one of its neighbors in $\mathcal{G}_{P'}(C^*)$ and call it $\gamma(u)$. We want to contract all edges $(u,\gamma(u))$. A refinement $P^*$ of $P'$ is created as follows. For each vertex $u \in Q$ and for each vertex $v \in P'^{-1}(u)$, we set $P^*(v) = \gamma(u)$. For all other vertices $v \in V$, we set $P^*(v) = P'(v)$. One can easily verify that the resulting mapping $P^*$ is a good partition for $C^*$.

Our next step is to argue that the length of the cycles in $P^*(C^*)$ is at most $\ell-1$. We first note that the fourth property in the statement of Lemma \ref{lemma:MainStep} ensures that every cycle $\mathfrak{c} \in C^*$ that is not a self-loop contains a vertex $u \in Q$. This vertex is well-contractible in $\mathcal G_{P'}(C^\star)$---due to the third property---and our construction above contracts an edge in $\mathfrak{c}$ incident to this vertex.

Therefore, to complete the proof, we only need to argue that if \RW$(\mathcal{G}_{P^*}(C^*),t)$ finds an odd-parity cycle with probability $\Omega_{\eps}(1)$, then the probability that \RW$(\mathcal{G}_{P}(C),3t)$ finds an odd-parity cycle is also $\Omega_{\eps}(1)$.

Each edge $(u,v)$ in $\mathcal{G}_{P^*}(C^*)$ belongs to a cycle $\mathfrak c$ in $P^*(C^*)$. By our construction, a given edge $(u,v)$ was either already present in $P(\mathfrak{c})$, or is a result of a contraction and there is a vertex $q \in Q$ such that $(u,q)$ and $(q,v)$ are edges of $P(\mathfrak{c})$. Note that in either case, the parity of the corresponding path from $u$ to $v$ is maintained. The probability to move from $u$ to $v$ via edge $(u,v)$ in $\mathcal{G}_{P^*}(C^*)$ is $1/\deg_{\mathcal{G}_{P^*}(C^*)}(u)$. If $(u,v)$ was already present in $\mathcal{G}_{P}(C)$, then the probability to move from $u$ to $v$ via edge $(u,v)$ in $\mathcal{G}_{P}(C)$ is $1/\deg_{\mathcal{G}_{P}(C)}(u)$. The case that $(u,v)$ corresponds to $(u,q)$ and $(q,v)$ in $\mathcal{G}_{P}(C)$ is more complicated since parallel edges become relevant. Since $q$ is a well-contractible vertex in $\mathcal G_P(C')$, we know that all cycles in $P(C')$ that pass through $q$ also go through $u$ and $v$, and all contain copies of the edges $(u,q)$ and $(q,v)$. Let us fix a copy of $(u,q)$ and call it $e$. In this case, the probability to move in $\mathcal{G}_{P}(C)$ from $u$ to $v$ through edge $e$ in two steps is
$1/(2\deg_{\mathcal{G}_{P}(C)}(u))$,
since the probability to move along $e$ is $1/\deg_{\mathcal{G}_{P}(C)}(u)$ and the probability to take a copy of $(q,v)$ is $\frac12$.
If we combine these arguments with the fifth property specified in Lemma \ref{lemma:MainStep}, we conclude that the probability to move from $u$ to $v$ along edge $(u,v)$ in $\mathcal{G}_{P^*}(C^*)$ differs from the probability to do the corresponding move in $\mathcal{G}_{P}(C)$ by at most a factor of $O_\epsilon(1)$. We further conclude inductively that if a random walk in $\mathcal{G}_{P^*}(C^*)$ moves in $t = O_{\epsilon}(1)$ steps from vertex $u$ to vertex $v$ with probability $p$, then the same movement happens in $\mathcal{G}_{P}(C)$ with probability $\Omega_{\epsilon}(p)$. Furthermore, by definition of well-contractible vertices, if the walk in $\mathcal{G}_{P^*}(C^*)$ contains a cycle of odd parity, then so does the corresponding walk in $\mathcal{G}_{P}(C)$.

It remains to address the probability of choosing $u$ as the starting vertex. Since we may contract many vertices into $u$ during our construction, the probability of choosing $u$ as a starting vertex in $\mathcal{G}_{P^*}(C^*)$ can be significantly greater than the probability of choosing $u$ in $\mathcal{G}_{P}(C)$. To this end, recall that the probability to choose $u$ as a starting vertex is $p' = |P^*{}^{-1}(u)| / |V|$. We note that with probability $\Omega_{\eps}(p')$ we sample either $u$ or a vertex $q\in Q$ that is contracted into $u$, i.e., a vertex $q$ with $\gamma(q) = u$. By the definition of well-contractible vertices, the probability to move from $q$ to $u$ in the first step of the random walk is at least $\frac12$. Summing up over all starting vertices (including $u$), we obtain that we end up at $u$ (either in step 1 or 2 of the random walk) with probability at least $p'/2$. Thus, for every walk that happens in $\mathcal{G}_{P^*}(C^*)$ with probability $p\cdot p'$ there is a (set of) walks in $\mathcal{G}_{P}(C)$ such that one of them happens with probability $\Omega_{\epsilon}(p \cdot p')$. Furthermore, if the walk in $\mathcal{G}_{P^*}(C^*)$ contains an odd-parity cycle then so does the walk in $\mathcal{G}_{P}(C)$. For different walks in $\mathcal{G}_{P^*}(C^*)$, the sets of corresponding walks in $\mathcal{G}_{P}(C)$ are disjoint. Thus, we do not double count and our result follows by observing that the walk in $\mathcal{G}_{P}(C)$ has length at most $2t+1 \le 3t$.
\qed

\junk{
\section{Proof of Lemma \ref{lemma:ExistenceOfH}}
Let $C$ be the final set of cycles $C_k$ obtained by from our reduction.
We claim that the lemma \ref{lemma:ExistenceOfH} is satisfied by graph $G(C_k)$. This can be seen by starting the reduction with $C_1=C$ and observing
that Lemmas \ref{lemma:small-vertices}, \ref{lemma:contractible-vertices}, and \ref{lemma:MainStep} are satisfied without modifications of the set of cycles $C$. Thus, during the reduction
the set $C$ does not change and at the end establishes the properties of Lemma \ref{lemma:ExistenceOfH} via Lemmas \ref{lemma:second_reduction} and \ref{lemma:Selfloops}.
}

\section{Extending the analysis to minor-free graphs}
\label{sec:minor-free}

While throughout the paper we focused on testing bipartiteness of planar graphs, our techniques can easily be extended to any class of minor-free graphs. Recall that a graph $H$ is called a \emph{minor} of a graph $G$ if $H$ can be obtained from $G$ via  a sequence of vertex and edge deletions and edge contractions. For any graph $H$, a graph $G$ is called \emph{$H$-minor-free} if $H$ is not a minor of $G$. (For example, by Kuratowski's Theorem, a graph is planar if and only if it is $K_{3,3}$-minor-free and $K_5$-minor-free.)

Let us fix a graph $H$ and consider the input graph $G$ to be an $H$-minor-free graph. We now argue now that entire analysis presented in the previous sections easily extends to testing bipartiteness of $G$. The key observation is that our analysis in Sections \ref{sec:preliminaries}--\ref{sec:proof-of-lemma:second_reduction} relies only on the following properties of planar graphs:
\begin{enumerate}[\it (i)]
\item the number of edges in a planar graph is $O(n)$, where $n$ is the number of vertices (Fact \ref{fact:planar_limited_edges}),
\item every minor of a planar graph is planar (Fact \ref{fact:planar_minor_planar}),
\item a direct implication of the Klein-Plotkin-Rao theorem for planar graphs (Lemma \ref{lemma:partition-into-conn-comp-of-small-diam}).
\end{enumerate}
The first two properties hold for any class of $H$-minor-free-graphs (that is, the second property would be that every minor of an $H$-minor-free-graph is $H$-minor-free). Since the Klein-Plotkin-Rao theorem holds for any minor-free graph as well (cf.~\cite{KPR93}), so does a version of Lemma \ref{lemma:partition-into-conn-comp-of-small-diam} with a slightly different constant hidden by the big $O$ notation. Therefore, we can proceed with nearly identical  analysis for $H$-minor-free graphs and arrive at the following version of Theorem \ref{thm:main-bipartiteness}.

\begin{theorem}
\label{thm:main-bipartiteness-minor-free}
Let $H$ be a fixed graph. There are positive functions $f$ and $g$ such that for any $H$-minor-free-graph $G$:
\begin{itemize}
\item if $G$ is bipartite, then \RBE$(G,\eps)$ accepts $G$, and
\item if $G$ is $\eps$-far from bipartite, then \RBE$(G,\eps)$ rejects $G$ with probability at least $0.99$.
\end{itemize}
\end{theorem}


\section{Conclusions}

In this paper we proved that bipartiteness is testable in constant time for arbitrary planar graphs. Our result was proven via a new type of analysis of random walks in planar graphs. Our analysis easily carries over to classes of graphs defined by arbitrary fixed forbidden minors.

This is merely the first step that poses the following main question:
\begin{center}
What graph properties can be tested in constant time in minor-free graphs?
\end{center}

Going through the analysis of the paper we obtain a running time of $2^{2^{\poly(1/\eps)}}$. While we did not try to optimize it, it seems that our technique
requires at least an exponential running time. An interesting open question is whether one can get polynomial or pseudopolynomial running time in $1/\epsilon$.
This seems to require significantly new techniques. In particular, to the best of our knowledge, it is not known how to obtain polynomial running time even for bounded-degree planar graphs.



\newcommand{\Proc}{Proceedings of the~}
\newcommand{\ALENEX}{Workshop on Algorithm Engineering and Experiments (ALENEX)}
\newcommand{\BEATCS}{Bulletin of the European Association for Theoretical Computer Science (BEATCS)}
\newcommand{\CCCG}{Canadian Conference on Computational Geometry (CCCG)}
\newcommand{\CIAC}{Italian Conference on Algorithms and Complexity (CIAC)}
\newcommand{\COCOON}{Annual International Computing Combinatorics Conference (COCOON)}
\newcommand{\COLT}{Annual Conference on Learning Theory (COLT)}
\newcommand{\COMPGEOM}{Annual ACM Symposium on Computational Geometry}
\newcommand{\DCGEOM}{Discrete \& Computational Geometry}
\newcommand{\DISC}{International Symposium on Distributed Computing (DISC)}
\newcommand{\ECCC}{Electronic Colloquium on Computational Complexity (ECCC)}
\newcommand{\ESA}{Annual European Symposium on Algorithms (ESA)}
\newcommand{\FOCS}{IEEE Symposium on Foundations of Computer Science (FOCS)}
\newcommand{\FSTTCS}{Foundations of Software Technology and Theoretical Computer Science (FSTTCS)}
\newcommand{\ICALP}{Annual International Colloquium on Automata, Languages and Programming (ICALP)}
\newcommand{\ICCCN}{IEEE International Conference on Computer Communications and Networks (ICCCN)}
\newcommand{\ICDCS}{International Conference on Distributed Computing Systems (ICDCS)}
\newcommand{\IJCGA}{International Journal of Computational Geometry and Applications}
\newcommand{\INFOCOM}{IEEE INFOCOM}
\newcommand{\IPCO}{International Integer Programming and Combinatorial Optimization Conference (IPCO)}
\newcommand{\ISAAC}{International Symposium on Algorithms and Computation (ISAAC)}
\newcommand{\ISTCS}{Israel Symposium on Theory of Computing and Systems (ISTCS)}
\newcommand{\JACM}{Journal of the ACM}
\newcommand{\LNCS}{Lecture Notes in Computer Science}
\newcommand{\PODS}{ACM SIGMOD Symposium on Principles of Database Systems (PODS)}
\newcommand{\RANDOM}{International Workshop on Randomization and Approximation Techniques in Computer Science (RANDOM)}
\newcommand{\RSA}{Random Structures and Algorithms}
\newcommand{\SICOMP}{SIAM Journal on Computing}
\newcommand{\SODA}{Annual ACM-SIAM Symposium on Discrete Algorithms (SODA)}
\newcommand{\SPAA}{Annual ACM Symposium on Parallel Algorithms and Architectures (SPAA)}
\newcommand{\STACS}{Annual Symposium on Theoretical Aspects of Computer Science (STACS)}
\newcommand{\STOC}{Annual ACM Symposium on Theory of Computing (STOC)}
\newcommand{\SWAT}{Scandinavian Workshop on Algorithm Theory (SWAT)}
\newcommand{\TALG}{ACM Transactions on Algorithms}
\newcommand{\UAI}{Conference on Uncertainty in Artificial Intelligence (UAI)}
\newcommand{\WADS}{Workshop on Algorithms and Data Structures (WADS)}
\newcommand{\TCS}{Theory of Computing Systems}

\renewcommand{\Proc}{{\rm In} Proceedings of the~}


\end{document}